\documentclass[english]{article}
\usepackage[T1]{fontenc}
\usepackage[latin9]{inputenc}
\usepackage{amsmath, amssymb, amsfonts}
\usepackage{amstext}
\usepackage{amsthm}
\usepackage{esint}
\usepackage{graphicx,color}
\usepackage{comment}
\usepackage{geometry}
\newcommand{\Tr}{\mathrm{Tr}}

\newcommand{\mc}[1]{\mathcal{#1}}

\newcommand{\ud}{\,\mathrm{d}}

\newcommand{\RR}{\mathbb{R}}
\newcommand{\CC}{\mathbb{C}}

\global\long\def\ve{\varepsilon}
\global\long\def\R{\mathbb{R}}
\global\long\def\Rn{\mathbb{R}^{d}}
\global\long\def\Rd{\mathbb{R}^{d}}

\global\long\def\vp{\varphi}
\global\long\def\ra{\rightarrow}

\global\long\def\dom{\mathrm{dom}\,}
\global\long\def\intdom{\mathrm{int}\,\mathrm{dom}\,}
\global\long\def\Tr{\mathrm{Tr}}

\numberwithin{equation}{section}
\numberwithin{figure}{section}
\newtheorem{thm}{\protect\theoremname}
\newtheorem{lem}[thm]{\protect\lemmaname}
\newtheorem{rem}[thm]{\protect\remarkname}
\newtheorem{prop}[thm]{\protect\propositionname}

\newtheorem{assumption}[thm]{Assumption}

\numberwithin{thm}{section}

\providecommand{\corollaryname}{Corollary}
\providecommand{\lemmaname}{Lemma}
\providecommand{\propositionname}{Proposition}
\providecommand{\remarkname}{Remark}
\providecommand{\theoremname}{Theorem}

\allowdisplaybreaks

\title{Sparsity pattern of the self-energy for \\ classical and quantum impurity
problems}
\author{
Lin Lin\thanks{Department of Mathematics, University of California, Berkeley, and Computational Research Division, Lawrence Berkeley National Laboratory, Berkeley, CA 94720. Email: \texttt{linlin@math.berkeley.edu}}
\and Michael Lindsey\thanks{Department of Mathematics, Courant Institute of Mathematical Sciences, New York University, New York, NY 10012. Email: \texttt{michael.lindsey@cims.nyu.edu}}
}
\date{}

\begin{document}
\maketitle

\begin{abstract}
  We prove that for various impurity models, in both classical and
  quantum settings, the self-energy matrix is a sparse matrix with a
  sparsity pattern determined by the impurity sites. 
  In the quantum setting, such a sparsity pattern has been known since
  Feynman. Indeed, it underlies several numerical methods for solving
  impurity problems, as well as many approaches to more general quantum many-body problems, such as the dynamical mean field theory. The sparsity pattern is easily motivated by a formal perturbative expansion using Feynman diagrams. However, to the extent of our knowledge, a
  rigorous proof has not appeared in the literature. In the classical 
  setting, analogous considerations lead to a perhaps less-known result, i.e., that the precision
  matrix of a Gibbs measure of a certain kind differs only by a sparse matrix
  from the precision matrix of a corresponding Gaussian measure. Our argument for this result mainly involves elementary algebraic manipulations and is in particular non-perturbative.  Nonetheless, the proof can be robustly
  adapted to various settings of interest in physics, including
  quantum systems (both fermionic and bosonic) at zero and finite temperature, non-equilibrium
  systems, and superconducting systems.
\end{abstract}

\pagestyle{myheadings}
\thispagestyle{plain}

\section{Introduction}
\label{sec:intro}

Consider the second-moment matrix $G\in \RR^{d\times d}$ of a Gibbs measure 
defined by a Hamiltonian $H  :\R^d \ra \R$, 
i.e., 
\begin{equation}
  G = \frac{1}{Z} \int_{\RR^{d}} xx^T e^{-H(x)} dx.
  \label{eqn:green}
\end{equation}
Here the partition function 
\[
Z=\int_{\RR^{d}} e^{-H(x)} dx
\]
is the appropriate normalization factor.

We will write $H$ in the form $H = H_0 + U$, 
where $H_0 = \frac12 x^T A x$ is a quadratic form.
Assume that $A$ and $U$ are 
such that both $Z$ and $G$ are finite. Via the analogy with quantum many-body physics 
that will be discussed below, we refer to $U$ as the \textit{interacting} part of 
the Hamiltonian, or simply the \textit{interaction}. 
Meanwhile $H_0$ represents the \textit{non-interacting} part.

If $U \equiv 0$ 
and $A$ is a positive definite matrix, then immediately we have
$G=A^{-1}$. One seeks a generalization of this fact to the case in which $U(x)$ depends only on a subset of the variables. 
We refer to this setting as the (classical) \emph{impurity model}, by analogy to the quantum impurity model 
to be discussed below.
Perhaps surprisingly, we have the following result:

\begin{prop}\label{thm:sparsegibbs}
  Let $p\leq d$, and let $A\in \RR^{d\times d}$ be a symmetric matrix
  whose lower-right $(d-p)\times(d-p)$ block is positive definite. Let
  $U:\RR^{d} \rightarrow \RR$ be a function that depends only on its
  first $p$ arguments, i.e., $U(x)=U_1 (x_1,\ldots,x_p)$ for some
  $U_1:\RR^p \rightarrow \RR$, and assume that $U_1$ satisfies
  sufficient growth conditions such that that the Gibbs measure with density proportional to $e^{-\frac12 x^{T} A x - U(x)}$ has finite second-order moments. Then, with $G$ defined as in \eqref{eqn:green}, 
  \begin{equation*}
    \Sigma := A - G^{-1} = \left(\begin{array}{cc}
 \Sigma_p & 0 \\
 0 & 0
\end{array}\right),
    \label{}
  \end{equation*}
  where $\Sigma_{p}\in \RR^{p\times p}$ is a symmetric matrix.
\end{prop}

In fact, Proposition \ref{thm:sparsegibbs} can be generalized by considering an arbitrary measure $d\mu_1 (\mathbf{x}_1)$ of sufficient 
decay in the place of $e^{-U_1(\mathbf{x}_1)}\, d \mathbf{x}_1$, where we denote
$\mathbf{x}_1 = (x_1,\ldots,x_p)^{T}$ and 
$\mathbf{x}_2 = (x_{p+1},\ldots,x_d)^{T}$. In this setting the partition function is defined 
\[
  Z =  \int_{\R^{p}} \int_{\R^{d-p}}  e^{-\frac12 x^T A x}\, d\mathbf{x}_2 \, d\mu_1 (\mathbf{x}_1),
\]
and the Green's function is defined accordingly.
The case
\[
\mu_1 (\mathbf{x}_1) = e^{-\sum_{i,j=1}^p J_{ij} x_i x_j} \sum_{\mathbf{\sigma}\in \{-1,1\}^p} \delta(\,\cdot\, - \mathbf{\sigma}\,)
\]
defines a notion of a classical impurity model for spin systems, in which a spin system 
is coupled to a Gaussian `bath.' For such a spin impurity model, we can assume
without loss of generality  that the upper-left $p\times p$ block 
of $A$ is zero, and the ensemble is specified by the partition function 
\[
  Z = \sum_{\mathbf{\sigma}\in \{-1,1\}^p} e^{-\frac12 \sum_{i,j=1}^p J_{ij} \sigma_i \sigma_j} 
 \int_{\R^{d-p}}   e^{- \frac{1}{2} y^T A_{22} y - y^T A_{21} \sigma} 
  \, dy,
\]
where $A_{21}$ and $A_{22}$ denote the appropriate blocks of $A$.
We will stick to the original setting, in which the impurity is specified by a function $U_1$, to 
emphasize the analogy with the setting of the quantum many-body problem, but we comment 
that the proof of 
Proposition \ref{thm:sparsegibbs} is exactly the same in this broader context.

In statistics, $G^{-1}$ is sometimes called the precision matrix. In our setting, if 
$A$ is positive definite and $U\equiv 0$, then $A$ is the precision matrix 
of the distribution in question. Hence Proposition~\ref{thm:sparsegibbs} states that
the difference of the precision matrices in the `interacting' and `non-interacting' settings, namely $A-G^{-1}$, is a
\textit{sparse} matrix if the interaction $U$ only depends on a subset of variables.
The proof of the theorem is non-perturbative, and in fact $A$ need not be positive definite (though, 
when $U$ is independent of the last $d-p$ variables, the lower-right $(d-p)\times(d-p)$ block of $A$ 
must be positive definite to ensure that $e^{-\frac12 x^{T} A x - U(x)}$ is normalizable).
To the best of our knowledge, other than from the
perspective of the Luttinger-Ward formalism to be discussed
later~\cite{LuttingerWard1960,LinLindsey2018,LinLindsey_diagram1,LinLindsey_diagram2},
this basic linear-algebraic fact about Gibbs measures was not
previously present in the literature.

As a matter of fact, we first observed a result of this type in a more complex
setting, namely that of quantum impurity problems at zero temperature (as we shall discuss
below, the analogous result is also true at finite temperature). Consider
the Hamiltonian, denoted by $\hat{H}$, for a system of interacting fermions or bosons. Throughout 
we shall distinguish the cases of fermions and bosons via a parameter $\zeta$ given by $\zeta = -1$ 
in the case of fermions and $\zeta = +1$ in the case of bosons. In the second-quantized
representation~\cite{FetterWalecka2003}, $\hat{H}$ can be
generally written as $\hat{H}= \hat{H}_{0} + \hat{U}$, where 
\begin{equation}
  \hat{H}_{0}=\sum_{i,j=1}^{d} h_{ij} a_{i}^{\dagger} a_{j}
  \label{}
\end{equation}
is viewed as the Hamiltonian for a system of non-interacting fermions or bosons.
Here $a_{i}^{\dagger},a_{j}$ are called the creation and annihilation
operators, respectively, and $h\in \CC^{d\times d}$ is a Hermitian
matrix (in Appendix \ref{sec:secondq} we provide a brief introduction of 
the second-quantized representation).

Meanwhile, $\hat{U}$ is the interacting part of the Hamiltonian. Although $\hat{U}$ 
can be far more general, usually we have in mind the two-body interaction
\begin{equation}
  \hat{U} = \sum_{i,j,k,l} (ij\vert U \vert kl) a^{\dagger}_{i}
  a^{\dagger}_{j} a_{l} a_{k}.
  \label{}
\end{equation}
In this case, if there exists $p<d$ so that $(ij\vert U \vert kl) \ne
0$ only if $i,j,k,l\in \left\{ 1,\ldots,p \right\}$, then we call the Hamiltonian
$\hat{H}$ an impurity Hamiltonian. More generally, 
we say that $\hat{H}$ is an impurity Hamiltonian if $\hat{U}$ can be written 
as a polynomial of the creation and annihilation operators $a_i^\dagger$ and $a_i$ 
for $i=1,\ldots,p$ and is particle-number-conserving (see Appendix \ref{sec:secondq} for details).

The quantum impurity problem arose as a model for magnetic impurities in metals~\cite{Anderson}. More recently, the impurity problem has assumed a role of central importance in the dynamical mean field theory (DMFT)~\cite{GeorgesKotliarKrauthEtAl1996,KotliarSavrasovHauleEtAl2006} and its
extensions~\cite{WolfGoMcCullochEtAl2015,LichtensteinKatsnelson2000}, which all in fact concern \emph{general} (i.e., non-impurity) quantum systems and in particular have distinguished themselves in the study of
strongly correlated fermionic systems that are difficult to treat by other means. As a `quantum embedding' method, DMFT considers a partition of  physical sites of a system into relatively small and localized fragments. For each fragment, DMFT defines an impurity problem in which the sites $i=1,\ldots, p$ correspond to the fragment and the sites $i > p$ constitute a `non-interacting bath,' meant to effectively mimic the effect of the environment on the fragment. The bath sites themselves are virtual in that they do not correspond to physical sites in the environment and are only meant to reproduce the environment's effect on the fragment. DMFT determines the effective impurity problems via a self-consistency condition by which global observables are forced to be compatible with local observables computed via the impurity problems. In an ongoing work, we use the sparsity result for the self-energy of the impurity problem to prove that the algorithmic DMFT loop that achieves this self-consistency is mathematically well-defined.

The key assumption underlying the self-consistency condition in DMFT is a block-diagonal ansatz for the global self-energy, with blocks specified by the self-energy matrices for each fragment. Computationally, the sparsity result implies that the self-energy can be recovered in terms of observables measured \emph{only} on the fragment part of each impurity problem. This observation plays a central role in numerical
algorithms for solving the quantum impurity problem, such as quantum
Monte Carlo (QMC) method~\cite{GullMillisLichtensteinEtAl2011}, where in fact the size of the bath may be thought of as effectively infinite.

  At a glance there is
no connection between this impurity Hamiltonian and the type of Gibbs
measure discussed earlier. Nonetheless, we claim that there is an analogy under
which $h$ maps to $A$, the single-particle Green's function
of the quantum many-body problem maps to $G$, and the self-energy matrix associated with
the Green's function maps to $\Sigma$. Then the counterpart of
Proposition~\ref{thm:sparsegibbs} can be stated in words as:
\textit{the self-energy matrix of a quantum impurity problem is a sparse
matrix, with nonzero entries only on the block associated with the impurity sites.}

The connection between
the classical impurity model and the quantum impurity problem
can be understood formally by writing quantum Green's functions 
in terms of the  coherent state path integral~\cite{NegeleOrland1988}, which formally 
resembles a Gibbs measure. We remark that in the case of 
fermions, the resemblance should be noted with special 
caution because the coherent state path integral involves 
Grassmann integrals. In this sense, the setting of
Proposition~\ref{thm:sparsegibbs} can indeed be understood as the 
`classical impurity problem.' 
Unlike the corresponding result for Gibbs measures, the quantum result has been well-known in the
quantum physics literature since Feynman and Vernon in
1963~\cite{FeynmanVernon1963} at the latest.
Again somewhat surprisingly, this important statement is to the
best of our knowledge a `folk theorem,' in that we cannot find
a rigorous proof of this result in the literature.  

In this paper, we fill this gap by providing rigorous proofs of the sparsity of the self-energy matrix of fermionic and bosonic quantum impurity problems at both zero and finite temperature (Theorems~\ref{thm:zeroTemp} and~\ref{thm:finiteTemp}, respectively). We will also cover the non-equilibrium setting (Theorem~\ref{thm:noneq}) via the consideration of arbitrary contour-ordered Green's functions, as well as the anomalous setting (Theorem~\ref{thm:anom}), which
is relevant to superconductivity. Excellent introductions to the 
non-equilibrium and anomalous formalisms can be found in \cite{StefanucciVanLeeuwen2013, BlaizotRipka1985}, 
respectively.

For concreteness, we outline here the setting of zero temperature and fixed particle number $N$, as well as our result in this setting.
Let $\big\vert \Psi_0^{(N)} \big\rangle$ denote a normalized $N$-particle ground state of $\hat{H}$, and 
let the corresponding eigenvalue be $E_0^{(N)}$.\footnote{The fact that $\left(E_0^{(N)},\big\vert \Psi_0^{(N)} \big\rangle\right)$ is the eigenpair corresponding to the ground state plays no role in the proof; it can be replaced by any other eigenpair of $\hat{H}$, and the statement remains valid. It is only natural due to physical reasons to consider  the ground state.} 
Then in this setting, the single-particle Green's function can be understood as a 
rational function $G: \mathbb{C} \ra \mathbb{C}^{d\times d}$ defined by 
$G(z) = G^{+}(z) + G^{-}(z)$, where 
$G^{\pm}$ are themselves rational functions\footnote{Usually $G^{\pm}$ carry the extra information that their poles are viewed as being located infinitesimally below/above the real axis. The choices that yield the `time-ordered' Green's function are described in Appendix \ref{sec:zeroGreens}. However, this extra information is irrelevant for the purpose of our results.} defined by 
\begin{equation*}
\begin{aligned}
G_{ij}^{+} (z) := & \ \big\langle \Psi_0^{(N)} \big\vert a_i \frac{1}{z - (\hat{H} - E_0^{(N)}) } a_j^{\dagger}  \big\vert  \Psi_0^{(N)} \big\rangle \\
G_{ij}^{-} (z) := & \ -\zeta \big\langle \Psi_0^{(N)} \big\vert a_j^{\dagger} \frac{1}{z + (\hat{H} - E_0^{(N)})} a_i  \big\vert  \Psi_0^{(N)} \big\rangle.
\end{aligned}
\end{equation*}
The self-energy is the rational function $\Sigma: \mathbb{C} \ra \mathbb{C}^{d\times d}$ defined by 
\[
\Sigma(z) := z - h - G(z)^{-1}.
\]
As we recover via Remark \ref{rem:nonInt} below, $z-h$ is in fact the inverse of the 
non-interacting Green's function, so this self-energy is defined analogously to the 
classical self-energy of Proposition \ref{thm:sparsegibbs}.
The reader should consult Appendix \ref{sec:zero} for further details and justification of these definitions. Then our main result on the sparsity pattern of the self-energy in this setting is stated as follows:

\begin{thm}
\label{thm:zeroTemp}
Suppose that 
$\hat{H}$ is an impurity Hamiltonian, with 
a fragment specified by the indices $1,\ldots,p$. Then the self-energy $\Sigma: \mathbb{C} \ra \mathbb{C}^{d\times d}$
is (up to the resolution of removable discontinuities) of the form 
  \begin{equation*}
    \Sigma(z) = \left(\begin{array}{cc}
 \Sigma_p (z) & 0 \\
 0 & 0
\end{array}\right).
  \end{equation*}
\end{thm}

\begin{rem}
\label{rem:nonInt}
Observe that if the fragment is of size zero, i.e., $p=0$, then we are in the non-interacting setting, and
 Theorem \ref{thm:zeroTemp} implies that $\Sigma(z) \equiv 0$, i.e., 
that $G(z) = (z-h)^{-1}$. Thus we recover a proof of the formula for the non-interacting Green's function.
%Usually this formula is proved by assuming, via a canonical transformation, that $h$ is diagonal and then performing explicit computations~\cite{FetterWalecka2003}.  
\end{rem}
Analogous results in other settings will be stated in the main text.

We hope that this work will have pedagogical value, especially to the mathematical audience 
unfamiliar with the physics literature. Since it is difficult to find standard 
references in the mathematics literature that are appropriate to our setting, 
we have included appendices to put our results on firm footing.
Via the appendices, we have also sought to make the work self-contained within reason, 
providing in particular some brief introduction 
to the theory of Green's functions, both fermionic and bosonic, in the zero-temperature, 
finite-temperature, non-equilibrium, and anomalous settings.
In all of these settings, the impurity model with 
$p=0$ is precisely the non-interacting model, and our results on the sparsity pattern of the 
self-energy, applied in this special case, yield formulas for the non-interacting Green's functions. In the 
non-equilibrium setting especially, such a formula seems to be non-trivial to establish by other means. 
Readers new to the subject may find this presentation of the non-interacting Green's functions, as well 
as its embedding into a unified perspective, to be appealing in its own right.

\subsection*{Other perspectives:}

We discuss several other ways of understanding the sparsity pattern 
of the self-energy for impurity problems. 
First, we remark by considering the coherent-state path integral
representation~\cite{NegeleOrland1988} (in any of the quantum settings discussed in this paper), 
one can formally view the quantum many-body ensemble as a Gibbs measure. The proof 
of Proposition \ref{thm:sparsegibbs} can be mimicked in these settings at the formal 
level to derive the appropriate sparsity results, but we omit such
formal manipulations here.

Secondly, the sparsity pattern can be most intuitively understood via the Feynman diagrammatic 
expansion, which provides another viewpoint on the formal unification of the 
classical and quantum settings.
Indeed, due to the connection
between the classical setting of Gibbs measures and the coherent state path integral, 
we limit our discussion the case of Gibbs models here
for simplicity.  We do not provide here a self-contained introduction to the diagrammatic 
expansion; instead we refer readers
to~\cite{AmitMartin-Mayor2005,NegeleOrland1988,LinLindsey_diagram1} for a more
detailed description.

As before, define the partition function 
\begin{equation}
  Z = \int_{\RR^{d}} e^{-\frac12 x^{T} A x - \varepsilon U(x)}\ud x,
  \label{eqn:partition}
\end{equation}
where $A$ is a positive definite matrix and where we have introduced the
parameter $\varepsilon>0$
as a prefactor for the interaction (referred to as the coupling constant). Then formally we may apply Taylor
expansion for $e^{-\varepsilon U(x)}$ to obtain a series expansion for
$Z$, as in
\begin{equation}
  Z = \int_{\RR^{d}} \sum_{n=0}^{\infty} \frac{\varepsilon^{n}}{n!} (-U(x))^{n}
  e^{-\frac12 x^{T} A x}\ud x \sim
  \sum_{n=0}^{\infty} \frac{\varepsilon^{n}}{n!} \int_{\RR^{d}} (-U(x))^{n}
  e^{-\frac12 x^{T} A x}\ud x,
  \label{eqn:Ztaylor}
\end{equation}
where the `$\sim$' is meant to indicate that the series is valid only in the asymptotic sense.

Assuming $U(x)$ is a polynomial of $x$, then each term on the right hand
side of Eq.~\eqref{eqn:Ztaylor} requires the evaluation of a possibly large, but 
finite, number of moments 
of a Gaussian distribution. The expansion can be organized in terms of Feynman diagrams.

Feynman diagrammatic expansions can also be obtained for
$G$ and $\Sigma$. In particular, the self-energy diagrams
are truncated, one-particle irreducible Feynman
diagrams~\cite{LinLindsey_diagram1}.  To be concrete, one can keep in mind the 
quartic interaction
\begin{equation}
  U(x) = \frac{1}{8} \sum_{i,j=1}^{N} v_{ij} x_{i}^2 x_{j}^2,
  \label{eqn:Uterm}
\end{equation}
which mimics the two-body Coulomb interaction of quantum many-body physics. 
Here $v$ is a symmetric positive definite matrix. 
In order to specify an impurity problem with fragment specified by indices $1,\ldots,p$ we take 
$v_{ij}=0$ if $i>p$ or $j>p$.  Then, it can readily be read from the diagrammatic expansion of 
$\Sigma$ as in \cite{LinLindsey_diagram1} that for each term in the expansion of $\Sigma_{ij}$, the
corresponding matrix element is nonzero only if $1\le i,j\le p$. This observation 
suggests that the self-energy matrix $\Sigma$, as the infinite sum of all of these 
terms, should follow the same sparsity pattern.  We remark that the above
diagrammatic argument can be applied to Gibbs models with rather general interaction form 
$U(x)$, as well as in the quantum many-body setting.
, where the diagrammatic series 
can be derived directly in the second-quantized representation or via the coherent state path integral.  

The major caveat to this argument is that the Feynman diagrammatic
expansion often has zero radius of 
convergence and maintains validity only in the asymptotic sense. This is the case
at least for the Gibbs models as well as bosonic systems. Hence the
sparsity for each term of the expansion does not necessarily imply that
the same is true of the self-energy itself when $\varepsilon$ is positive.  Even when the series
does converge (such as for fermionic systems with finitely many states), the
convergence radius may only be finite. Bootstrapping a positive radius of convergence 
via resummation or analytic continuation
arguments~\cite{NegeleOrland1988} is one possible route to proving the sparsity result 
in such a setting, though the details seem to be cumbersome and the proof is not as simple 
or general as others considered above.

Next, we discuss a route to the sparsity of the self-energy matrix via the 
so-called Luttinger-Ward formalism~\cite{LuttingerWard1960}, which expresses the self-energy 
as a functional derivative 
\begin{equation}
  \Sigma = \frac{\delta \Phi[G]}{\delta G}.
  \label{eqn:sigmaLW}
\end{equation}
Here $\Phi[G]$ is a functional of the Green's function, called the
Luttinger-Ward functional. Recently, for the Gibbs model, we have proved~\cite{LinLindsey2018,LinLindsey_diagram2} 
that $\Phi[G]$ is a well-defined for positive semidefinite $G$. In particular, 
we have established a \textit{projection rule}, which states that for
the classical impurity problem when $U(x)=U_{p}(x_{1},\ldots,x_{p})$, we
have
\begin{equation}
  \Phi[G] = \Phi_{p}[G_{p}].
  \label{eqn:projectionrule}
\end{equation}
Here $G_{p}$ is the upper-left $p \times p$ block of $G$, and
$\Phi_{p}$ is the Luttinger-Ward functional for the $p$-dimensional
model.  Combining Eq.~\eqref{eqn:sigmaLW}
and~\eqref{eqn:projectionrule}, one immediately obtains the sparsity pattern
for $\Sigma$.

However, establishing the existence of the Luttinger-Ward functional $\Phi[G]$ and 
its projection rule require a significant amount of work, and
the rigorous proof is so far only applicable to the Gibbs model. In fact, the
very existence of the Luttinger-Ward functional 
fermionic systems has been challenged over the past few
years~\cite{KozikFerreroGeorges2015,Elder2014,GunnarssonRohringerSchaeferEtAl2017}. 
Although the Luttinger-Ward perspective offers additional insight, 
the direct proofs provided in this paper are at this point more generally applicable, 
and certainly much simpler.

Finally, our results in the non-equilibrium setting should be compared to 
those in a recent work~\cite{CorneanMoldoveanuPillet2018}, in which \emph{advanced}/\emph{retarded} non-equilibrium
self-energies are rigorously constructed in the case of fermions. Though
not noted explicitly in the work, the appropriate sparsity results for these 
quantities can be seen to follow from the construction itself. 
By contrast, our non-equilibrium sparsity result concerns the \emph{contour-ordered} 
self-energy (for both fermions and bosons) and in particular recovers sparsity results for the advanced/retarded Green's functions. 
Moreover, our result holds for arbitrary contour. However, we do not actually 
construct the contour-ordered self-energy, but rather phrase our sparsity 
result in terms of operators that we suggestively name `$G\Sigma$' and `$\Sigma G$.'\footnote{The 
reader will find that from the point of view of this paper, 
these objects can be thought of as more natural than 
the non-equilibrium self-energy itself, and indeed all of our sparsity results are proved 
by considering their analogs.}
In so doing we sidestep a considerable analytical challenge such as that encountered in 
\cite{CorneanMoldoveanuPillet2018}. Thus our result can be viewed as trying to parsimoniously illustrate the broadest 
possible formal picture of sparsity results for the self-energy, rather than focusing on the analytical 
question of the construction of the self-energy itself. Incidentally, in our view 
a rigorous construction of the contour-ordered self-energy (for
arbitrary contour) seems to be an interesting and non-trivial matter.

\subsection*{Outline of the paper:}

This paper is organized as follows. We use the classical impurity
problem as a motivating example and prove
Proposition~\ref{thm:sparsegibbs} in section~\ref{sec:gibbs}. Section~\ref{sec:quantumImpurity} 
treats the quantum many-body case, including the settings of fermions and bosons in the 
equilibrium setting at zero and finite temperature, as well as the non-equilibrium setting specified 
by an arbitrary contour in the complex plane and the anomalous setting relevant to 
superconductivity.

Finally, in Appendix \ref{sec:secondq} we record self-contained background on second quantization. 
In Appendix \ref{sec:zero} we discuss the zero-temperature ensemble for fermions and bosons and 
the construction of the frequency representation of Green's functions in this setting. In Appendix \ref{sec:finite} 
we do the same for the finite-temperature ensemble. Some efforts must be made here to deal with 
analytical issues in the bosonic case, where the Fock space is infinite-dimensional, even for finitely 
many states. In Appendix \ref{sec:non-equilibrium} we discuss the technical conditions needed 
to define the appropriate objects in the bosonic non-equilibrium setting and provide some
background on main non-equilibrium setting of interest, specified by the Kadanoff-Baym contour.

\subsection*{Acknowledgments:} 

This work was partially supported by the Department of Energy under
Grant No. DE-SC0017867, No. DE-AC02-05CH11231, by the Air Force Office
of Scientific Research under award number FA9550-18-1-0095  (L. L.),
by the National Science Foundation Graduate Research Fellowship Program
under grant DGE-1106400 (M. L.), and by the National Science Foundation under Award No. 1903031 (M.L.).  We thank Garnet Chan, Jianfeng Lu,
Nicolai Reshetikhin, Reinhold Schneider, and Lexing Ying for helpful discussions.

\section{The classical impurity problem (Proof of Proposition \ref{thm:sparsegibbs})}\label{sec:gibbs}

We now embark upon the proof of Proposition \ref{thm:sparsegibbs} stated above.

Recall the definitions: 
\[
Z=\int_{\Rn}e^{-\frac{1}{2}x^{T}Ax-U(x)}\,dx,\ \ G=\frac{1}{Z}\int_{\Rn}xx^T e^{-\frac{1}{2}x^{T}Ax-U(x)}\,dx,
\]
 where the interaction $U$ only depends on the first $p\leq d$ variables.
Let $q=d-p$. It is not hard to see that $G$ is positive definite, hence 
invertible.

We will indicate the blocks of $A$ via 
\[
A=\left(\begin{array}{cc}
A_{11} & A_{12}\\
A_{21} & A_{22}
\end{array}\right), 
\]
where the upper-left block is $p\times p$. 
For various integrals considered below to be convergent, we
will require that $A_{22} \succ0$.
More generally, we adopt the notation that for any $d\times d$ matrix 
$M$, the notation $M_{21}$ indicates the 
lower-left block of $M$ (with respect to the above block structure), etc.

Then for the theorem, we want to show that the self-energy $\Sigma:=A-G^{-1}$ satisfies 
$\Sigma_{12} = 0$, $\Sigma_{21} = 0$, and $\Sigma_{22} = 0$. In other words, 
we want to show that $(G^{-1})_{12}= A_{12}$, $(G^{-1})_{21}= A_{21}$, and 
$(G^{-1})_{22}= A_{22}$.
Since $G$ and $A$ are
symmetric, it suffices to show that
$(G^{-1})_{12}= A_{12}$ and 
$(G^{-1})_{22}= A_{22}$, i.e., that 
\[
\left(\begin{array}{cc}
 (G^{-1})_{12}\\
(G^{-1})_{22}
\end{array}\right) 
= \left(\begin{array}{cc}
A_{12}\\
A_{22}
\end{array}\right).
\]
Left-multiplying both sides by $G$ (invertible), we see that this is in turn 
equivalent to showing that $(GA)_{12} = 0_{\,p\times q}$ and $(GA)_{22} = I_q$.

In the following our notation will make use of the splitting 
\[
x=\left(\begin{array}{c}
x_1\\
x_2
\end{array}\right),
\]
where $x\in \R^d$, $x_1\in\R^{p}$, and $x_2\in\R^{q}$. (For notational 
convenience, we \emph{do not} use the notation $\mathbf{x}_i$ as in the introduction. 
In this section, we will make no reference to the individual entries of $x$, so the notation 
is clear.)
Then we can write $U(x)=U_1(x_1)$. Abusing notation slightly, we write $U_1 = U$. 

Roughly speaking, the goal is to `integrate out' the lower variables (i.e., 
the last $q$ variables). To this end, we expand $G$ as  
\[
G =\frac{1}{Z}\int_{\R^{p}}e^{-U(x_1)}\int_{\R^{q}} xx^T \exp\left[-\frac{1}{2}\left(\begin{array}{c}
x_1\\
x_2
\end{array}\right)^{T}\left(\begin{array}{cc}
A_{11} & A_{12}\\
A_{21} & A_{22}
\end{array}\right)\left(\begin{array}{c}
x_1\\
x_2
\end{array}\right)\right]\,dx_2\,dx_1.
\]
 Observe that 
\begin{equation}
\label{eqn:schurDerivation}
\begin{split}
& \left(\begin{array}{c}
x_1\\
x_2
\end{array}\right)^{T}\left(\begin{array}{cc}
A_{11} & A_{12}\\
A_{21} & A_{22}
\end{array}\right)\left(\begin{array}{c}
x_1\\
x_2
\end{array}\right) 
\\ & \quad\quad =\  \left(x_2+A_{22}^{-1}A_{21} x_1 \right)^{T} A_{22} \left(x_2+A_{22}^{-1}A_{21} x_1 \right)
 +\ x_1^T (A_{11}-A_{12} A_{22}^{-1} A_{21}) x_1, 
 \\ & \quad\quad = \ \left(x_2+A_{22}^{-1}A_{21} x_1 \right)^{T} A_{22} \left(x_2+A_{22}^{-1}A_{21} x_1 \right)
 +\ x_1^T A_{11}^{\mathrm{S}} x_1,
\end{split}
\end{equation}
where $A_{22}^{-1}$ is understood always to indicate $(A_{22})^{-1}$ and 
where we have defined the Schur complement 
\[
A_{11}^{\mathrm{S}} := A_{11}-A_{12} A_{22}^{-1} A_{21}.
\]
Then it follows that 
\begin{equation}
\label{eqn:bigG}
\begin{split}
&G \ =\
\frac{1}{Z}\int_{\R^{p}}e^{-\frac{1}{2}x_1^{T}A_{11}^{\mathrm{S}}
x_1-U(x_1)}\times\\
   & \quad\quad\quad\quad \int_{\R^{q}} xx^T \exp\left[-\frac{1}{2}\left(x_2+A_{22}^{-1}A_{21} x_1 \right)^{T} A_{22} \left(x_2+A_{22}^{-1}A_{21} x_1 \right)\right]\,dx_2\,dx_1.
\end{split}
\end{equation}
Recall that we want to show that $(GA)_{12} = 0$ and $(GA)_{22} = I_q$. 
Right-multiplying the integral in \eqref{eqn:bigG} by $A$, this
motivates computing the upper-right and upper-left blocks of 
$xx^T A$, as in
\[
(xx^T A)_{12} = x_1\ 
(x_1^T \ x_2^T)
\left(\begin{array}{cc}
A_{12}\\
A_{22}
\end{array}\right),
\quad
(xx^T A)_{22} = x_2\ 
(x_1^T \ x_2^T)
\left(\begin{array}{cc}
A_{12}\\
A_{22}
\end{array}\right).
\]
Now 
\[
(x_1^T \ x_2^T)
\left(\begin{array}{cc}
A_{12}\\
A_{22}
\end{array}\right) 
= 
x_1^T A_{12} + x_2^T A_{22}
=
(x_1^T A_{12} A_{22}^{-1} + x_2^T) A_{22}
= y_2^T A_{22},
\]
where 
we have defined a new variable $y_2 = x_2+A_{22}^{-1}A_{21} x_1$, 
so 
\begin{equation}
\label{eqn:change}
(xx^T A)_{12} = x_1 y_2^T A_{22}, 
\quad
(xx^T A)_{22} = x_2 y_2^T A_{22}.
\end{equation}
The remarkable thing is that $x_2$ only appears in the exponent 
in the inner integrand of \eqref{eqn:bigG} via the expression 
$x_2 + A_{22}^{-1} A_{21} x_1 = y_2$. This motivates us 
to eliminate $x_2$ from the second equation of 
\eqref{eqn:change} to obtain 
\begin{equation}
\label{eqn:change2}
(xx^T A)_{12} = x_1 y_2^T A_{22}, 
\quad
(xx^T A)_{22} = y_2 y_2^T A_{22} - A_{22}^{-1} A_{21} x_1 y_2^T A_{22}.
\end{equation}
Then consider the change of variables from 
$x_1, x_2$ to $x_1, y_2$, yielded by the linear transformation
\[
\left(\begin{array}{cc}
x_1 \\
y_2
\end{array}\right) = \left(\begin{array}{cc}
I_p & 0 \\
A_{22}^{-1} A_{21} & I_q
\end{array}\right)
\left(\begin{array}{cc}
x_1 \\
x_2
\end{array}\right).
\]
Since the Jacobian determinant of this transformation is one, 
it follows from \eqref{eqn:bigG} and 
\eqref{eqn:change2} that 
\begin{equation*}
(GA)_{12} \ =\  \frac{1}{Z}\int_{\R^{p}}e^{-\frac{1}{2}x_1^{T}A_{11}^{\mathrm{S}} x_1-U(x_1)} \, x_1 
\left( \int_{\R^{q}} y_2^T e^ {-\frac{1}{2} y_2^T A_{22} y_2 }\,dy_2 \right)  A_{22} \,dx_1.
\end{equation*}
But evidently the inner integrand is zero, so $(GA)_{12} = 0$, as desired.
It also follows from \eqref{eqn:bigG} and 
\eqref{eqn:change2} that 
\begin{equation*}
\begin{split}
&(GA)_{22} \ =\  \frac{1}{Z}\int_{\R^{p}}e^{-\frac{1}{2}x_1^{T}A_{11}^{\mathrm{S}} x_1-U(x_1)}  
\left( \int_{\R^{q}} y_2 y_2^T e^ {-\frac{1}{2} y_2^T A_{22} y_2 }\,dy_2 \right)  A_{22} \,dx_1 \\
& \quad\quad\quad\quad\  - \ 
\frac{1}{Z}\int_{\R^{p}}e^{-\frac{1}{2}x_1^{T}A_{11}^{\mathrm{S}} x_1-U(x_1)}  A_{22}^{-1} A_{21} x_1
\left( \int_{\R^{q}} y_2^T e^ {-\frac{1}{2} y_2^T A_{22} y_2 }\,dy_2 \right)  A_{22} \,dx_1.
\end{split}
\end{equation*}
The inner integrand in the second term of the last expression is once again zero. 
Meanwhile, the inner integrand of the first term yields 
$Z_2 A_{22}^{-1}$, where 
\[
Z_2 := \int_{\R^{q}} e^ {-\frac{1}{2} y_2^T A_{22} y_2} \,dy_2.
\]
Then we have established 
\[
(GA)_{22} = \frac{I_p}{Z} \int_{\R^{p}} \int_{\R^{q}} e^{-\frac{1}{2}x_1^{T}A_{11}^{\mathrm{S}} x_1
 - \frac{1}{2} y_2^T A_{22} y_2 -U(x_1)}  \,dy_2\,dx_1.
\]
Changing variables back to $x_1, x_2$ and recalling from 
\eqref{eqn:schurDerivation} that 
$x_1^T A_{11}^{\mathrm{S}} x_1 + y_2^T A_{22} y_2 = x^T A x$, we see that 
\[
(GA)_{22} = \frac{I_p}{Z} \int_{\R^{d}}e^{-\frac{1}{2}x^T A x -U(x)}  \,dx = I_p,
\]
which completes the proof. $\square$

\section{The quantum impurity problem}
\label{sec:quantumImpurity}

Our setting in this section is the Fock space 
$\mathcal{F}_{\zeta,d}$ of fermions ($\zeta = -1$) or bosons ($\zeta = +1$) 
with a finite number $d$ of states. The annihilation and creation operators are 
denoted $a_1,\ldots,a_d$ and $a_1^\dagger, \ldots, a_d^\dagger$, respectively. 
We refer the reader to Appendix \ref{sec:secondq} for further details of the construction 
of $\mathcal{F}_{\zeta,d}$ as well as other details of second quantization.
For convenience we shall let $a = (a_1, \ldots, a_d)^{T}$ denote the vector 
of annihilation operators, and accordingly $a^\dagger = (a_1^\dagger, \ldots, a_d^\dagger)$.

For now\footnote{In sections \ref{sec:arbitraryContour} and \ref{sec:anomalous} below, the notion of the Hamiltonian will be somewhat modified.} we consider a particle-number-conserving\footnote{See Appendix \ref{sec:secondq} for a details.} 
self-adjoint Hamiltonian $\hat{H}$ on $\mathcal{F}_{\zeta,d}$, and we write $\hat{H}$ of the form
\[
\hat{H} = \hat{H}_0 + \hat{U},
\]
where 
\[
\hat{H}_0 := a^\dagger h a = \sum_{i,j=1}^d h_{ij} a_i^\dagger a_j
\]
is the \emph{single-particle} (or \emph{non-interacting}) part of the Hamiltonian, 
specified by a Hermitian $d\times d$ matrix $h$, and $\hat{U}$ 
is the \emph{interacting} part, which is itself a self-adjoint operator on $\mathcal{F}_{\zeta,d}$ 
that conserves particle number.

In the case that $\hat{U}$ can be written as a polynomial of the $a_i^\dagger, a_i$ 
for $i=1,\ldots, p$, we say that $\hat{H}$ is an \emph{impurity Hamiltonian}, with 
a \emph{fragment} specified by the indices $1,\ldots,p$. The rest of the indices 
correspond to the \emph{environment}. In this case, since $\hat{U}$ conserves particle number, 
it follows that $\hat{U}$ commutes with $a_j$ and $a_j^\dagger$ for $j > p$.

Before proceeding, we state and prove a simple but useful lemma that will be used 
repeatedly throughout the following discussion.
\begin{lem}
\label{lem:key}
$[a^{\dagger}ha,a_{j}^{\dagger}] = \sum_{k=1}^d h_{kj}a_{k}^{\dagger}$ and 
$[a_j, a^{\dagger}ha] = \sum_{l=1}^d h_{jl}a_{l}$.
\end{lem}
\begin{proof}
Simply compute
\begin{eqnarray*}
(a^\dagger h a) a_j^\dagger & = & \sum_{k,l=1}^d h_{kl} a_k^\dagger a_l a_j^\dagger \\
& = & \sum_{k,l=1}^d h_{kl} a_k^\dagger (\zeta a_j^\dagger a_l + \delta_{jl}) \\ 
& = & \sum_{k,l=1}^d h_{kl} a_j^\dagger a_k^\dagger  a_l + \sum_{k=1}^d h_{kj} a_k^\dagger \\ 
& = & a_j^\dagger (a^\dagger h a) + \sum_{k=1}^d h_{kj} a_k^\dagger, 
\end{eqnarray*}
which proves the first statement of the lemma. Similarly, 
\begin{eqnarray*}
(a^\dagger h a) a_j & = &  \sum_{k,l=1}^d h_{kl} \zeta a_k^\dagger a_j a_l \\
& = & \sum_{k,l=1}^d h_{kl} (a_j a_k^\dagger - \delta_{jk})  a_l \\ 
& = & a_j (a^\dagger h a) - \sum_{l=1}^d h_{jl} a_l
\end{eqnarray*}
which proves the lemma.
\end{proof}

\subsection{Zero temperature (Proof of Theorem~\ref{thm:zeroTemp})}\label{sec:sparseZeroTemp}
We consider the setting of zero temperature and fixed particle number $N$. The reader should recall the notation for this setting introduced in section~\ref{sec:intro}. (Note that further background and detail is provided in Appendix \ref{sec:zeroGreens}.) Then we turn to the proof of Theorem~\ref{thm:zeroTemp} stated in section~\ref{sec:intro}.

\begin{proof}[Proof of Theorem~\ref{thm:zeroTemp}]
We can write $\hat{H} = \hat{H}_0 + \hat{U}$, 
where
$\hat{H}_0 = a^\dagger h a$ 
and  $\hat{U}$ commutes with $a_j$ and $a_j^\dagger$ for $j > p$
is the \emph{interacting} part, which is itself a self-adjoint operator on $\mathcal{F}_{\zeta,d}$ 
that conserves particle number.

It suffices to prove that the $j$-th column of $G(z) \Sigma(z) $ is zero for $j > p$ and that 
the $i$-th row of $\Sigma(z) G(z)$ is zero for $i > p$. We will only prove the first claim; 
the second follows by symmetric reasoning.

Now $ G(z) \Sigma(z) = z G(z) - G(z) h  - I_d$, so we want to show that $z G_{ij}(z) = [G(z) h]_{ij} + \delta_{ij}$ 
for $j > p$.

Then we compute, using the fact that $(\hat{H} - E_0^{(N)}) \big\vert \Psi_0^{(N)}\big\rangle = 0$,  
\begin{eqnarray*}
z G_{ij}^{+} (z) & = & \big\langle \Psi_0^{(N)} \big\vert a_i \frac{1}{z - (\hat{H} - E_0^{(N)}) } a_j^{\dagger}    ( z - (\hat{H} - E_0^{(N)}) ) \big\vert \Psi_0^{(N)} \big\rangle \\
& = & \big\langle \Psi_0^{(N)} \big\vert a_i \frac{1}{z - (\hat{H} - E_0^{(N)}) }  ( z - (\hat{H} - E_0^{(N)}) ) a_j^{\dagger}   \big\vert \Psi_0^{(N)} \big\rangle \\ 
& & \ \ \ + \  \ \big\langle \Psi_0^{(N)} \big\vert a_i \frac{1}{z - (\hat{H} - E_0^{(N)}) } [ a_j^{\dagger}, z - (\hat{H} - E_0^{(N)}) ]  \big\vert \Psi_0^{(N)} \big\rangle \\ 
& = & \big\langle \Psi_0^{(N)} \big\vert a_i  a_j^{\dagger}   \big\vert \Psi_0^{(N)} \big\rangle 
 + \big\langle \Psi_0^{(N)} \big\vert a_i \frac{1}{z - (\hat{H} - E_0^{(N)}) } [ a_j^{\dagger},  z - (\hat{H} - E_0^{(N)})  ]  \big\vert \Psi_0^{(N)} \big\rangle.
\end{eqnarray*}
Now 
\[
[ a_j^{\dagger},  z - (\hat{H} - E_0^{(N)}) ] = [ \hat{H}, a_j^{\dagger} ] = [a^\dagger h a, a_j^\dagger ] + [\hat{U},a_j^\dagger] =  \sum_{k=1}^d h_{kj}a_{k}^{\dagger},
\]
where we have used Lemma \ref{lem:key} as well as the fact that $j>p$ (so $[\hat{U},a_j^\dagger] = 0$).

Then it follows that 
\[
z G_{ij}^{+} (z) = \big\langle \Psi_0^{(N)} \big\vert a_i a_j^{\dagger}
\big\vert \Psi_0^{(N)} \big\rangle + [G^{+}(z) h]_{ij}.
\]

Similarly, we compute 
\begin{eqnarray*}
z G_{ij}^{-} (z) & = & -\zeta \big\langle \Psi_0^{(N)} \big\vert  ( z + (\hat{H} - E_0^{(N)}) )a_j^\dagger \frac{1}{z + (\hat{H} - E_0^{(N)}) } a_i   \big\vert \Psi_0^{(N)} \big\rangle \\
& = & -\zeta \big\langle \Psi_0^{(N)} \big\vert a_j^\dagger ( z + (\hat{H} - E_0^{(N)}) )\frac{1}{z + (\hat{H} - E_0^{(N)}) }   a_i   \big\vert \Psi_0^{(N)} \big\rangle \\ 
& & \ \ \ + \  \ (-\zeta)\big\langle \Psi_0^{(N)} \big\vert [ z + (\hat{H} - E_0^{(N)}), a_j^\dagger ] \frac{1}{z + (\hat{H} - E_0^{(N)}) }  a_i \big\vert \Psi_0^{(N)} \big\rangle \\ 
& = & -\zeta \big\langle \Psi_0^{(N)} \big\vert a_j^{\dagger}  a_i    \big\vert \Psi_0^{(N)} \big\rangle 
  -\zeta \big\langle \Psi_0^{(N)} \big\vert [ z + (\hat{H} - E_0^{(N)}), a_j^\dagger ]  \frac{1}{z + (\hat{H} - E_0^{(N)}) } a_i \big\vert \Psi_0^{(N)} \big\rangle.
\end{eqnarray*}
Now 
\[
[ z + (\hat{H} - E_0^{(N)}), a_j^\dagger ]  = [ \hat{H}, a_j^\dagger] =  \sum_{k=1}^d h_{kj}a_{k}^{\dagger}.
\]
Then it follows that 
\[
z G_{ij}^{-} (z) = -\zeta \big\langle \Psi_0^{(N)} \big\vert a_i a_j^{\dagger}   \big\vert \Psi_0^{(N)} \big\rangle + [G^{-}(z) h]_{ij}.
\]
Therefore 
\[
z G_{ij} (z) = [G(z) h]_{ij} + \big\langle \Psi_0^{(N)} \big\vert a_i a_j^{\dagger} - \zeta a_j^\dagger a_i  \big\vert \Psi_0^{(N)} \big\rangle = 
 [G(z) h]_{ij} + \delta_{ij},
\]
as was to be shown.
\end{proof}

\subsection{Finite temperature}\label{sec:sparseFiniteTemp}
Now we consider the setting of finite inverse temperature $\beta\in (0,\infty)$ and chemical potential 
$\mu \in \intdom Z$, where
$Z(\mu) = \Tr[e^{-\beta (\hat H - \mu \hat{N})}]$ and $\intdom Z$ is the interior of the effective domain $\mathrm{dom}\,Z = \{\mu \,: \, Z(\mu) < \infty\}$ (see Appendix \ref{sec:finite} for further details). Note 
that $\intdom Z$ is guaranteed to be non-empty under Assumption \ref{assumption:bosons}.

 We also let 
$\vert \Psi_m \rangle$ denote the normalized eigenstates of $\hat{H}$, where $m$ ranges from 
$0$ to $2^{d} - 1$ in the case of fermions and from $0$ to $\infty$ in the case of bosons.
In this setting, the single-particle Green's function can be understood as a 
rational function $G: \mathbb{C} \ra \mathbb{C}^{d\times d}$ defined by 
$G(z) = G^{+}(z) + G^{-}(z)$, where 
$G^{\pm}$ are themselves rational functions\footnote{The same comments as in section \ref{sec:sparseZeroTemp} apply here as well, though instead see Appendix \ref{sec:finiteGreens} for details relevant to this setting.} defined by 
\begin{equation*}
\begin{aligned}
G_{ij}^{+} (z) := & \ \frac{1}{Z} \sum_m e^{-\beta (E_m - \mu N_m)} \big\langle \Psi_m \big\vert a_i \frac{1}{z - (\hat{H}- E_m)} a_j^{\dagger}  \big\vert \Psi_m \big\rangle \\
G_{ij}^{-} (z) := & \ \frac{-\zeta}{Z} \sum_m e^{-\beta (E_m - \mu N_m)} \big\langle \Psi_m \big\vert a_j^{\dagger} \frac{1}{z + (\hat{H}- E_m) } a_i  \big\vert \Psi_m \big\rangle,
\end{aligned}
\end{equation*}
and these sums are absolutely convergent away from the poles. Here 
\[
Z = \Tr[e^{-\beta (\hat H - \mu \hat{N})}] = \sum_m e^{-\beta (E_m - \mu N_m)}.
\]
Once again the self-energy is the rational function $\Sigma: \mathbb{C} \ra \mathbb{C}^{d\times d}$ defined by 
\[
\Sigma(z) := z - h - G(z)^{-1}.
\]
The reader should consult Appendix \ref{sec:finite} for further details and justification of these definitions.

\begin{thm}
\label{thm:finiteTemp}
Suppose that 
$\hat{H}$ is an impurity Hamiltonian, with 
a fragment specified by the indices $1,\ldots,p$. Then the self-energy $\Sigma: \mathbb{C} \ra \mathbb{C}^{d\times d}$
is (up to the resolution of removable discontinuities) of the form 
  \begin{equation*}
    \Sigma(z) = \left(\begin{array}{cc}
 \Sigma_p (z) & 0 \\
 0 & 0
\end{array}\right).
  \end{equation*}
\end{thm}

\begin{rem}
Once again (cf. Remark \ref{rem:nonInt}), we recover in the non-interacting setting the formula 
$G(z) = (z-h)^{-1}$.
\end{rem}

\begin{rem}
\label{rem:matsubara}
There is a further object known as the \emph{Matsubara Green's function} \cite{NegeleOrland1988}, which in turn yields the Matsubara 
self-energy. Although it is not usually defined this way, the Matsubara Green's function can be shown to be obtained from the finite-temperature Green's function, as defined above, by restriction to points $i\omega_m +\mu$, where $\omega_m$ are the fermionic/bosonic \emph{Matsubara frequencies} \cite{NegeleOrland1988}. The Matsubara self-energy can be 
obtained from the finite-temperature self-energy defined above via similar restriction. Therefore Theorem \ref{thm:finiteTemp} implies the same sparsity pattern for the Matsubara self-energy.
\end{rem}

\begin{proof}
The proof is essentially the same as that of Theorem \ref{thm:zeroTemp}. Once again we want to show that the $j$-th column of $ G(z) \Sigma(z)$ is zero for $j > p$ and that 
the $i$-th row of $ \Sigma(z) G(z)$ is zero for $i > p$. We will only prove the first claim; 
the second follows by symmetric reasoning.

Define $G_m (z)$ by 
\[
G_{m,ij} (z) :=  \big\langle \Psi_m \big\vert a_i \frac{1}{z - (\hat{H}- E_m)} a_j^{\dagger}  \big\vert \Psi_m \big\rangle 
- \zeta \big\langle \Psi_m \big\vert a_j^{\dagger} \frac{1}{z + (\hat{H}- E_m) } a_i  \big\vert \Psi_m \big\rangle.
\]
Then by the same reasoning as in the proof of Theorem \ref{thm:zeroTemp} (with the roles of $E_0^{(N)}$ and $\vert \Psi_0^{(N)} \rangle$ 
played by $E_m$ and $\vert \Psi_m \rangle$, we find that 
\[
z G_{m,ij}(z) = [G_m(z) h]_{ij} + \delta_{ij}.
\]
Now $G(z) = \frac{1}{Z} \sum_m e^{-\beta (E_m - \mu N_m)} G_{m,ij} (z)$, so the desired result follows.
\end{proof}

\subsection{Arbitrary contour}
\label{sec:arbitraryContour}
There is a more general perspective in which the time-ordering 
operation used in Appendices
\ref{sec:zero} and \ref{sec:finite} to derive the Green's functions considered above
is generalized to an ordering operation on an arbitrary 
contour in the complex plane. 
This perspective adds significant value in the 
non-equilibrium setting, in which one considers a time-dependent Hamiltonian. For such time-dependent 
problems, passage to the frequency representation is not possible. Instead we consider kernels on the 
contour.

Let $\mathcal{C}$ denote a piecewise smooth contour in the complex plane (not necessarily closed). 
Technically one should think of $\mathcal{C}$ not as a subset of $\mathbb{C}$, but as a parametrized path, 
$\gamma : I \ra \mathbb{C}$, where $I = (s_0, s_1)$ is some interval. Then for $s,s' \in I$ with $s < s'$, we define $\mathcal{C}(s,s')$ to be the `sub-contour' defined by restriction of $\gamma$ to the interval $(s,s')$. If $s > s'$, we define $\mathcal{C}(s,s')$ 
to be the contour obtained from $\mathcal{C}(s',s)$ by reversing its orientation. 

Additionally let $\hat{H}(z)$ denote an operator-valued
function on a neighborhood of $\mathcal{C} = \gamma(I)$. 
Here $\hat{H}(z) = a^\dagger h(z) a + \hat{U}(z)$ is particle-number-conserving, and we say that $\hat{H}(z)$ is an impurity Hamiltonian with a fragment specified by indices 
$1,\ldots,p$ if, for every $z \in \mathcal{C}$, $\hat{U}(z)$ can be written as a polynomial of the $a_i^\dagger, a_i$ for $i=1,\ldots,p$. As above, since $\hat{U}(z)$ must conserve particle number, it follows that $\hat{U}(z)$ commutes 
with $a_j$ and $a_j^\dagger$ for $j> p$.
It is convenient to denote $z(s) := \gamma(s)$, and abusing notation 
slightly we will write $\hat{H}(s) = \hat{H}(z(s))$.

For simplicity, we assume that $\hat{H}(s)$ is piecewise continuous, which will ensure that various integrals are well-defined below. Since the Fock space 
is finite dimensional in the case of fermions, the meaning of this statement is unambiguous. In 
the case of bosons, note that since $\hat{H}(s)$ is particle-number-conserving, we can 
sensibly consider its restriction to each of the $N$-particle subspaces (see Appendix \ref{sec:secondq}), 
each of which is finite-dimensional. Then by the continuity of $\hat{H}(s)$ we mean the continuity 
of all of these restrictions individually.

Now define a (not necessarily unitary) evolution operator from contour time $s' \in I$ to $s\in I$ as the time-ordered 
exponential
\[
U(s, s') = \mathcal{T} \left\{ e^{ -i \int_{\mathcal{C}(s,s')} \hat{H}(z) \,dz} \right\}.
\]
This simply means that $U(s,s')$ is taken as the solution of the differential equation 
\begin{equation}
\label{eqn:dsU}
\partial_s U(s,s') = -i \,\dot{z}(s) \hat{H}(s) U(s,s'),\quad U(s',s') = \mathrm{Id}.
\end{equation}
This initial-value problem indeed admits a unique solution in the bosonic case 
because the ODE can be viewed as describing the evolution of an operator on
 each of the (finite-dimensional) $N$-particle subspace separately.

From this definition it follows that 
\[
U(s,s'') U(s'', s') = U(s,s') 
\]
for all $s,s',s'' \in I$ and moreover that
\begin{equation}
\label{eqn:dsU2}
\partial_{s'} U(s,s') = i \,\dot{z}(s) U(s,s') \hat{H}(s').
\end{equation}
Abusing notation slightly by pretending that we can invert $s = s(z)$, 
we can more cleanly write 
\[
\partial _z U(z,z') = -i \hat{H}(z) U(z,z'), \quad \partial_{z'} U(z,z') = i U(z,z') \hat{H}(z'), 
\]
where $\partial_z = (\dot{z}(s))^{-1} \partial_s$. We will sometimes adopt this notational convention, 
and the meaning should be clear from context. 

The following assumption is adopted to ensure that the 
Green's function can be defined in the bosonic case:
\begin{assumption}
\label{assumption:traceClass}
We assume that for all $s > s'$, $U(s,s')$ is a bounded operator. Moreover, 
we assume that there exists $s > s'$ such that the operator norm of the 
restriction of $U(s,s')$ to the $N$-particle subspace decays exponentially in $N$.
\end{assumption}

Define the partition function 
\[
Z = \Tr [U(s_1, s_0)].
\]
Note that Assumption \ref{assumption:traceClass} guarantees that $U(s_1,s_0)$ is trace class, so 
$Z$ is indeed well-defined.
In order to define our ensemble, we must be able to divide by $Z$. Hence we assume:
\begin{assumption}
\label{assumption:Znonzero}
$Z \neq 0$.
\end{assumption} 
We show in Appendix \ref{sec:non-equilibrium} 
how Assumptions \ref{assumption:traceClass} and \ref{assumption:Znonzero} 
are naturally satisfied in the major non-equilibrium setting of interest, which features 
the Kadanoff-Baym contour. 

Then we define `pseudo-Heisenberg' representations of the annihilation and creation 
operators via
\[
a_i(s) = U(s_0,s) a_i U(s,s_0), \quad a_i^\dagger(s) = U(s_0,s) a_i^\dagger U(s,s_0).
\]
The \emph{contour-ordered, single-body Green's function} (which we call the Green's function 
for short when the context is clear) is a function $G: I \times I \ra \mathbb{C}^{d \times d}$ defined by 
\[
G_{ij}(s,s') =  \frac{-i}{Z} \Tr \big[ \mathcal{T} \big\{ a_i(s) a_j^\dagger (s') \big\} U(s_1, s_0) \big],
\]
where $\mc{T}$ is the \emph{contour-ordering operator}, formally defined by 
\begin{equation*}
   \mathcal{T}\big\{ a_i(s)  a_j^\dagger (s') \big\} =
   \begin{cases}
    a_{i}(s)a_{j}^{\dagger}(s'), & s' < s\\
     \zeta a_{j}^{\dagger}(s') a_{i}(s), & s' \geq s.
  \end{cases}
  \label{}
\end{equation*}
In other words we can write $G = G^{+} + G^{-}$, where 
\[
i G_{ij}^+(s,s') = \frac{1}{Z} \Tr \big[  U(s_1, s) a_i U(s,s') a_j^\dagger U(s',s_0) \big] \theta(s-s')
\]
and
\[
i G_{ij}^-(s,s') = \frac{\zeta}{Z} \Tr \big[  U(s_1, s') a_j^\dagger U(s',s)  a_iU(s,s_0)  \big] (1 - \theta(s-s')).
\]
Here 
 \begin{equation*}
 \theta(s) := \begin{cases}
    1, & s > 0\\ 
     0, & s \leq 0.
  \end{cases}
 \end{equation*}
In the bosonic case, Assumption \ref{assumption:traceClass} guarantees 
that the traces needed for this definition do indeed exist.
For later reference, note that we can define a product of suitable functions $A,B: I \times I \ra \mathbb{C}^{d \times d}$ 
(with an appropriate notion of multiplicative inverse, at least formally) via
\[
(AB)(s,s') = \int_{s_0}^{s_1} A(s,s'')B(s'',s')\, \dot{z}(s'') \,ds'',
\]
chosen so that formally we have 
\[
(AB)(z,z') = \int_{z_0}^{z_1} A(z,z'')B(z'',z') \,dz''.
\]
Notice that the appropriate identity $\delta(z,z')$ is then given by $\delta(z,z') = (\dot{z}(s))^{-1} \delta(s-s')$. This last expression should be interpreted carefully as an integral operator $A(s,s')$ on $L^2(I)$. Indeed, $z(s)$ is generally only piecewise smooth, so $(\dot{z}(s))^{-1}$ may be ill-defined at finitely many points, but nonetheless the expression remains well-defined as such an operator.

We remark that the zero-temperature and Matsubara Green's functions discussed in 
section~\ref{sec:sparseZeroTemp} and Remark~\ref{rem:matsubara}, respectively, can be recovered as contour-ordered Green's functions. 
By contrast, the real-time Green's function at finite temperature
considered in section \ref{sec:sparseFiniteTemp} \emph{cannot} be recovered 
directly as a contour-ordered Green's function, though it can be obtained indirectly via analytic continuation of the Matsubara 
Green's function. For this reason, diagrammatic expansion techniques at finite temperature 
are limited to the Matsubara Green's function and must be carried over to the 
real-time Green's function via analytic continuation. For further details, see \cite{StefanucciVanLeeuwen2013}.

One now wants to define the self-energy as 
\[
\Sigma(z,z') = i \partial_z - h(z)\, \delta(z,z') - G^{-1} (z,z').
\]
However, this definition is not rigorous without further justification. 
Indeed, note that $G$ can be viewed as an integral operator on $L^2 (I)$, and under reasonable assumptions 
$G$ is Hilbert-Schmidt, hence in particular compact.
Therefore its inverse is guaranteed to be an 
unbounded operator, if it can be constructed. Formally, one expects that the $i (\dot{z}(s))^{-1} \partial_s$ 
in our definition of the self-energy
will cancel an analogous term in the formal inverse $G^{-1}$ and that the self-energy 
can be written as a sum of a static and dynamic part as
\[
\Sigma(s,s') = \Sigma_{\mathrm{stat}}\, \delta(s - s') + \Sigma_{\mathrm{dyn}} (s,s'), 
\]
where $\Sigma_{\mathrm{dyn}}$ is a properly defined integral operator.

In our view the mathematical construction of the self-energy seems to be a non-trivial 
matter, and we will sidestep it in this work. (By contrast, the construction in the equilibrium 
setting is more straightforward in the frequency domain; see Appendices \ref{sec:zero} and \ref{sec:finite}.)

How then to discuss the sparsity pattern of the self-energy? Observe that formally, we should 
have
\begin{equation*}
\begin{split}
(\Sigma G)(z,z') &= i\partial_z G (z,z') - h(z) G(z,z') -  I_d \,\delta(z, z')   \\
(G \Sigma) (z,z') & = - i\partial_{z'} G (z,z') - G(z,z') h(z') -  I_d\, \delta(z, z'),
\end{split}
\end{equation*}
or, more rigorously, 
\begin{equation}
\label{eqn:SigmaG}
\begin{split}
(\Sigma G)(s,s') &= i(\dot{z}(s))^{-1} \partial_s G (s,s') - h(s) G(s,s') -  I_d \,(\dot{z}(s))^{-1} \,\delta(s - s')   \\
(G \Sigma) (s,s') & = - i(\dot{z}(s'))^{-1} \partial_{s'} G (s,s') - G(s,s') h(s') -  I_d\, (\dot{z}(s))^{-1}\, \delta(s - s').
\end{split}
\end{equation}
Again observe that equality is not meant to be interpreted pointwise, but rather in the sense of integral operators, as $(\dot{z}(s))^{-1}$ may be ill-defined at finitely many points.

Now instead of constructing the self-energy, we can \emph{define} operators $\Sigma G$ and $G \Sigma$ via \eqref{eqn:SigmaG} (in the sense of distributions), 
with the `$\Sigma$' appearing here merely as a notation. 
Now the desired sparsity pattern of $\Sigma$ is formally equivalent to the statement that $[\Sigma G]_{ij} = 0$ 
(as a distribution on $I$) for $i > p$ and $[G \Sigma]_{ij} = 0$ for $j > p$.

\begin{thm}
\label{thm:noneq}
With notation and assumptions as in the preceding, if $\hat{H}(z)$ is an impurity Hamiltonian with a
fragment specified by the indices $1,\ldots,p$, then
 $[\Sigma G]_{ij} = 0$ for $i>p$ and 
$[G \Sigma]_{ij} = 0$ for $j>p$.
\end{thm}
\begin{rem}
In the non-interacting setting $p=0$, we recover the formulas 
\[
i \partial_z G (z,z') - h(z)\,G(z,z') =  I_d \,\delta(z, z'), \quad  - i \partial_{z'} G (z,z') - G(z,z')\,h(z') =  I_d \,\delta(z, z'),
\]
where we have abused notation slightly in the manner described above. 
These formulas seem to be non-trivial to establish by any other means. By contrast with the equilibrium case, 
this formula cannot be established simply via a canonical transformation because it may not be 
possible to simultaneously diagonalize the
$h(z)$ for all $z$. In fact, in \cite{StefanucciVanLeeuwen2013}, the non-interacting Green's function is \emph{defined} via this 
formula (subject to certain boundary conditions) 
and shown to give the appropriate perturbation theory within the Martin-Schwinger hierarchy.
\end{rem}

\begin{proof}
We prove only the first statement, as the second follows from similar arguments. Recall
\[
i G_{ij}^{+}(s,s')= \frac{1}{Z}\Tr\left[U(s_{1},s)a_{i}U(s,s')a_{j}^{\dagger}U(s',s_{0})\right]\theta(s-s').
\]
Then compute, using Eqs. \eqref{eqn:dsU} and
\eqref{eqn:dsU2}, 
\begin{eqnarray*}
i  (\dot{z}(s))^{-1} \partial_{s}G_{ij}^{+}(s,s') & = & \frac{i}{Z}\Tr\left[U(s_{1},s)\hat{H}(s)a_{i}U(s,s')a_{j}^{\dagger}U(s',s_{0})\right]\theta(s-s')\\
 &  & \ \ -\ \frac{i}{Z}\Tr\left[U(s_{1},s)a_{i}\hat{H}(s)U(s,s')a_{j}^{\dagger}U(s',s_{0})\right]\theta(s-s')\\
 &  & \ \ +\ (\dot{z}(s))^{-1}\frac{1}{Z}\Tr\left[U(s_{1},s)a_{i}a_{j}^{\dagger}U(s,s_{0})\right]\delta(s-s')\\
 & = & \frac{-i}{Z}\Tr\left[U(s_{1},s)[a_{i},\hat{H}(s)]U(s,s')a_{j}^{\dagger}U(s',s_{0})\right]\theta(s-s')\\
 &  & \ \ +\ (\dot{z}(s))^{-1}\frac{1}{Z}\Tr\left[U(s_{1},s)a_{i}a_{j}^{\dagger}U(s,s_{0})\right]\delta(s-s').
\end{eqnarray*}
Now for $i>p$, $[a_i, \hat{U}(s)] = 0$, so $[a_i, \hat{H}(s)] = [a_i, a^\dagger h(s) a] = \sum_{l=1}^d h_{il}(s) a_l$, by 
Lemma \ref{lem:key}.
Therefore
\begin{eqnarray*}
i  (\dot{z}(s))^{-1} \partial_{s}G_{ij}^{+}(s,s') 
 & = & \frac{-i}{Z}\sum_{l=1}^{d}h_{il}(s) \Tr\left[U(s_{1},s)a_{l}U(s,s')a_{j}^{\dagger}U(s',s_{0})\right]\theta(s-s')\\
 &  & \ \ +\ (\dot{z}(s))^{-1}\frac{1}{Z}\Tr\left[U(s_{1},s)a_{i}a_{j}^{\dagger}U(s,s_{0})\right]\delta(s-s')\\
 & = & [h(s) G^{+}(s,s')]_{ij}+(\dot{z}(s))^{-1}\frac{1}{Z}\Tr\left[U(s_{1},s)a_{i}a_{j}^{\dagger}U(s,s_{0})\right]\delta(s-s').
\end{eqnarray*}
Similarly, 
\begin{eqnarray*}
i (\dot{z}(s))^{-1} \partial_{s}G_{ij}^{-}(s,s') & = & \frac{i\zeta}{Z}\Tr\left[U(s_{1},s')a_{j}^{\dagger}U(s',s)\hat{H}(s)a_{i}U(s,s_{0})\right](1-\theta(s-s'))\\
 &  & \ \ -\ \frac{i\zeta}{Z}\left[U(s_{1},s')a_{j}^{\dagger}U(s',s)a_{i}\hat{H}(s)U(s,s_{0})\right](1-\theta(s-s'))\\
 &  & \ \ - \ (\dot{z}(s))^{-1} \frac{\zeta}{Z}\Tr\left[U(s_{1},s)a_{j}^{\dagger}a_{i}U(s,s_{0})\right]\delta(s-s')\\
 & = & \frac{-i\zeta}{Z}\Tr\left[U(s_{1},s)[a_{i},\hat{H}(s)]U(s,s')a_{j}^{\dagger}U(s',s_{0})\right]\theta(s-s')\\
 &  & \ \ -\ \frac{\zeta}{Z}\Tr\left[U(s_{1},s)a_{j}^{\dagger}a_{i}U(s,s_{0})\right]\delta(s-s')\\
 & = & \frac{-i\zeta}{Z}\sum_{l=1}^{d}h_{il}(s)\Tr\left[U(s_{1},s) a_{l} U(s,s')a_{j}^{\dagger}U(s',s_{0})\right]\theta(s-s')\\
 &  & \ \ -\ (\dot{z}(s))^{-1} \frac{\zeta}{Z}\Tr\left[U(s_{1},s)a_{j}^{\dagger}a_{i}U(s,s_{0})\right]\delta(s-s')\\
 & = & [h(s) G^{-}(s,s')]_{ij}- (\dot{z}(s))^{-1} \frac{\zeta}{Z}\Tr\left[U(s_{1},s)a_{j}^{\dagger}a_{i}U(s,s_{0})\right]\delta(s-s').
\end{eqnarray*}
Therefore, since $G = G^+ + G^-$, we have 
\begin{eqnarray*}
i (\dot{z}(s))^{-1} \partial_{s}G_{ij}(s,s') & = & [h(s) G(s,s')]_{ij} \\ & & \quad + \ \frac{1}{Z}\Tr\left[U(s_{1},s)(a_{i}a_{j}^{\dagger}-\zeta a_{j}^{\dagger}a_{i})U(s,s_{0})\right] (\dot{z}(s))^{-1} \delta(s-s')\\
 & = & [h(s) G(s,s')]_{ij}+\delta_{ij}\,(\dot{z}(s))^{-1} \delta(s-s'),
\end{eqnarray*}
which completes the proof.
\end{proof}

\subsection{Anomalous setting}
\label{sec:anomalous}
Finally we will consider a sparsity result for the self-energy of \emph{anomalous} impurity problems. These are impurity problems 
in which the Hamiltonian does not conserve particle number. Since the anomalous setting is of most interest for the 
study of superconductivity in \emph{fermions}, we will restrict our attention to the fermionic setting. This allows us to avoid 
some further analytic difficulty since our rigorous definitions in the
bosonic case (in which the Fock space is infinite-dimensional) relied on
particle number conservation. It also eases the notational burden to
keep track of $\zeta$ to distinguish the bosonic and
fermionic systems. In order to
simply illustrate the points that are novel to this setting, we 
further restrict our attention to the zero-temperature equilibrium setting.

Now consider a self-adjoint Hamiltonian $\hat{H}$ on the fermionic Fock
space $\mathcal{F}_{-1,d}$, and we write $\hat{H}$ of the form
\[
\hat{H} = \hat{H}_0 + \hat{U},
\]
where 
\[
\hat{H}_0 := \hat{H}_{\mathrm{NA}} + \hat{H}_{\mathrm{A}} + \hat{H}_{\mathrm{A}}^\dagger
\]
is the single-particle part of the Hamiltonian (no longer particle-number-conserving), specified by its 
non-anomalous and anomalous parts 
\[
\hat{H}_{\mathrm{NA}} := \sum_{i,j=1}^d h_{ij} a_i^\dagger a_j ,\quad \hat{H}_{\mathrm{A}} := \frac{1}{2} \sum_{i,j=1}^d \Delta_{ij} a_i^\dagger a_j^\dagger.
\]
Therefore, up to a scalar multiple of the identity operator, $\hat{H}_0$ is given by 
\[
\left(\begin{array}{cc}
a \\
a^\dagger
\end{array}\right)^\dagger
 \left(\begin{array}{cc}
{h} & {\Delta}\\
-\overline{{\Delta}} & -\overline{{h}}
\end{array}\right)
\left(\begin{array}{cc}
a \\
a^\dagger
\end{array}\right),
\]
where we have abused notation slightly by using $a$ to indicate both a row and a column vector of operators.
Without loss of generality we assume that $\Delta = (\Delta_{ij})$ is a complex antisymmetric matrix. (Note that 
then $-\overline{\Delta} = \Delta^\dagger$, and since $h$ is Hermitian, $-\overline{h} = -h^{T}$.)
Meanwhile, the interacting part $\hat{U}$ is itself a self-adjoint operator on $\mathcal{F}_{-1,d}$, and we 
demand that it can be written as an \emph{even} polynomial of the
creation and annihilation operators, which includes the
particle-number-conserving $\hat{U}$ as a sub-case.
In the case that $\hat{U}$ can be written as a polynomial of the $a_i^\dagger, a_i$ 
for $i=1,\ldots, p$, we say that $\hat{H}$ is an \emph{anomalous impurity Hamiltonian}, with 
a fragment specified by the indices $1,\ldots,p$. As in earlier settings, the rest of the indices 
correspond to the environment. Note that the evenness of the polynomial specifying $\hat{U}$ guarantees that $\hat{U}$ commutes with $a_j$ and $a_j^\dagger$ for $j > p$.

Now define the following Green's functions: 
\[
G_{ij}^{\mathrm{hp}}(z):=G_{ij}^{\mathrm{hp},+}(z)+G_{ij}^{\mathrm{hp},-}(z):=\langle\Phi_{0}\vert a_{i}\frac{1}{z-(\hat{H}-E_{0})}a_{j}^{\dagger}\vert\Phi_{0}\rangle+\langle\Phi_{0}\vert a_{j}^{\dagger}\frac{1}{z+(\hat{H}-E_{0})}a_{i}\vert\Phi_{0}\rangle
\]
 
\[
G_{ij}^{\mathrm{pp}}(z):=G_{ij}^{\mathrm{pp},+}(z)+G_{ij}^{\mathrm{pp},-}(z):=\langle\Phi_{0}\vert a_{i}^{\dagger}\frac{1}{z-(\hat{H}-E_{0})}a_{j}^{\dagger}\vert\Phi_{0}\rangle+\langle\Phi_{0}\vert a_{j}^{\dagger}\frac{1}{z+(\hat{H}-E_{0})}a_{i}^{\dagger}\vert\Phi_{0}\rangle
\]
 
\[
G_{ij}^{\mathrm{hh}}(z):=G_{ij}^{\mathrm{hh},+}(z)+G_{ij}^{\mathrm{hh},-}(z):=\langle\Phi_{0}\vert a_{i}\frac{1}{z-(\hat{H}-E_{0})}a_{j}\vert\Phi_{0}\rangle+\langle\Phi_{0}\vert a_{j}\frac{1}{z+(\hat{H}-E_{0})}a_{i}\vert\Phi_{0}\rangle
\]

\[
G_{ij}^{\mathrm{ph}}(z):=G_{ij}^{\mathrm{ph},+}(z)+G_{ij}^{\mathrm{ph},-}(z):=\langle\Phi_{0}\vert a_{i}^{\dagger}\frac{1}{z-(\hat{H}-E_{0})}a_{j}\vert\Phi_{0}\rangle+\langle\Phi_{0}\vert a_{j}\frac{1}{z+(\hat{H}-E_{0})}a_{i}^{\dagger}\vert\Phi_{0}\rangle,
\]
 where $\vert\Phi_{0}\rangle$ is the ground state of $\hat{H}$
and $E_{0}$ is the ground-state energy. The superscripts $p$ and
$h$ stands for `particle' and `hole', respectively~\cite{BlaizotRipka1985}, so
$G^{\mathrm{hh}}$ is called the hole-hole Green's function,
$G^{\mathrm{ph}}$ is the particle-hole Green's function, etc.

Notice that the last two
Green's functions are actually redundant because $G^{\mathrm{ph}}(z)=-[G^{\mathrm{hp}}(-z)]^{T}$
and $G^{\mathrm{hh}}(z)=[G^{\mathrm{pp}}(\overline{z})]^{\dagger}$.
We can further define the anomalous Green's function by 
\[
\mathbf{G}(z):=\left(\begin{array}{cc}
G^{\mathrm{hp}}(z) & G^{\mathrm{hh}}(z)\\
G^{\mathrm{pp}}(z) & G^{\mathrm{ph}}(z)
\end{array}\right)
\]
and the anomalous self-energy by 
\[
\mathbf{\Sigma}(z) := z- \left(\begin{array}{cc}
{h} & {\Delta}\\
-\overline{{\Delta}} & -\overline{{h}}
\end{array}\right) - \mathbf{G}(z)^{-1}.
\]

In fact we will show the following result: 
\begin{thm}
\label{thm:anom}
Suppose that 
$\hat{H}$ is an anomalous impurity Hamiltonian, with 
a fragment specified by the indices $1,\ldots,p$. Then the anomalous self-energy $\mathbf{\Sigma}: \mathbb{C} \ra \mathbb{C}^{d\times d}$
is (up to the resolution of removable discontinuities) of the form 
  \begin{equation*}
    \mathbf{\Sigma}(z) =
    \left(\begin{array}{cccc}
 \Sigma_p^{\mathrm{hp}} (z) & 0 & \Sigma_p^{\mathrm{hh}} (z) & 0  \\
 0 & 0 & 0 & 0 \\
  \Sigma_p^{\mathrm{pp}} (z) & 0 & \Sigma_p^{\mathrm{ph}} (z) & 0\\ 
   0 & 0 & 0 & 0
\end{array}\right).
  \end{equation*}
\end{thm}

\begin{rem}
Note that in the case $p=0$ we recover the formula 
\[
\mathbf{G}(z) = \left[ z- \left(\begin{array}{cc}
{h} & {\Delta}\\
-\overline{{\Delta}} & -\overline{{h}}
\end{array}\right) \right]^{-1}
\]
for the non-interacting anomalous Green's function.
.\end{rem}

Recall from Lemma \ref{lem:key} that 
\[
[\hat{H}_{\mathrm{NA}},a_{j}^{\dagger}]=\sum_{\gamma}h_{kj}a_{k}^{\dagger}
\]
 and 
\[
[\hat{H}_{\mathrm{NA}},a_{j}]=-\sum_{\gamma}h_{jk}a_{k}^{\dagger}.
\]
Before proceeding with the proof of Theorem \ref{thm:anom}, we supplement this result with a further simple lemma: 
\begin{lem}
\label{lem:key2}
Let $\hat{H}_{\mathrm{A}} = \frac{1}{2} \sum_{i,j=1}^d \Delta_{ij}
a_i^\dagger a_j^\dagger$ with $\Delta = (\Delta_{ij})$ being a complex
antisymmetric matrix. Then 
\[
[\hat{H}_{\mathrm{A}},a_{j}^{\dagger}]=0,\quad[\hat{H}_{\mathrm{A}}^{\dagger},a_{j}]=0, 
\quad [\hat{H}_{\mathrm{A}},a_{j}]=\sum_{k}\Delta_{kj}a_{k}^{\dagger},
\quad [\hat{H}_{\mathrm{A}}^{\dagger},a_{j}^{\dagger}]=\sum_{k}\overline{\Delta}_{jk}a_{k}.
\]
\end{lem}
\begin{proof}
The first two identities are obvious, and the fourth follows from the third by taking Hermitian conjugates and using the 
antisymmetry of $\Delta$.
To see the claimed third identity, simply compute 
\begin{eqnarray*}
\hat{H}_{\mathrm{A}}a_{j} & = & \frac{1}{2}\sum_{ik}\Delta_{ik}a_{i}^{\dagger}a_{k}^{\dagger}a_{j}\\
 & = & \frac{1}{2}\sum_{ik}\Delta_{ik}a_{i}^{\dagger}\delta_{jk}-\frac{1}{2}\sum_{ik}\Delta_{ik}a_{i}^{\dagger}a_{j}a_{k}^{\dagger}\\
 & = & \frac{1}{2}\sum_{i}\Delta_{ij}a_{i}^{\dagger}-\frac{1}{2}\sum_{ik}\Delta_{ik}\delta_{ij}a_{k}^{\dagger}+\frac{1}{2}\sum_{ik}\Delta_{ik}a_{j}a_{i}^{\dagger}a_{k}^{\dagger}\\
 & = & \frac{1}{2}\sum_{k}\Delta_{kj}a_{k}^{\dagger}-\frac{1}{2}\sum_{k}\Delta_{jk}a_{k}^{\dagger}+a_{j}\hat{H}_{\mathrm{A}}\\
 & = & \sum_{k}\Delta_{kj}a_{k}^{\dagger}+a_{j}\hat{H}_{\mathrm{A}}.
\end{eqnarray*}
\end{proof}

\begin{proof}
(Of Theorem \ref{thm:anom}.)
Throughout we will often use $\langle\,\cdot\,\rangle$
to indicate the expectation $\langle\Phi_{0}\vert\,\cdot\,\vert\Phi_{0}\rangle$.

Now it suffices to show the following sparsity pattern
\[
\mathbf{G}(z) \mathbf{\Sigma(z)} =  
    \left(\begin{array}{cccc}
 * & 0 & * & 0  \\
* & 0 & * & 0  \\
* & 0 & * & 0  \\
* & 0 & * & 0 
\end{array}\right), 
\quad 
 \mathbf{\Sigma(z)} \mathbf{G}(z)=  
    \left(\begin{array}{cccc}
 * & * & * & *  \\
0 & 0 & 0 & 0  \\
* & * & * & *  \\
0 & 0 & 0 & 0 
\end{array}\right). 
\]
We will only prove the first of these claims; the other follows by similar reasoning. Note that this first 
claim is equivalent to the fact that each of the following equalities
holds \emph{along the last $d-p$ columns}:
\[
G^{\mathrm{hp}}[z-{h}]-G^{\mathrm{hh}}{\Delta}^{\dagger}={I_d},
\]
\[
-G^{\mathrm{hp}}{\Delta}+G^{\mathrm{hh}}[z+{h}^{T}]={0},
\]
\[
G^{\mathrm{pp}}[z-{h}]-G^{\mathrm{ph}}{\Delta}^{\dagger}={0},
\]
\[
-G^{\mathrm{pp}}{\Delta}+G^{\mathrm{ph}}[z+{h}^{T}]={I_d}.
\]

Now we begin the computations. In the following we assume that $j>p$.
Since $(\hat{H}-E_{0})\vert\Phi_{0}\rangle={0}$, we have 
\begin{eqnarray*}
zG_{ij}^{\mathrm{hp},+}(z) & = & \langle\Phi_{0}\vert a_{i}\frac{1}{z-(\hat{H}-E_{0})}a_{j}^{\dagger}(z-(\hat{H}-E_{0}))\vert\Phi_{0}\rangle\\
 & = & \langle a_{i}a_{j}^{\dagger}\rangle+\langle a_{i}\frac{1}{z-(\hat{H}-E_{0})}[\hat{H},a_{j}^{\dagger}]\rangle\\
 & = & \langle a_{i}a_{j}^{\dagger}\rangle+\langle a_{i}\frac{1}{z-(\hat{H}-E_{0})}[\hat{H}_{\mathrm{NA}},a_{j}^{\dagger}]\rangle+\langle a_{i}\frac{1}{z-(\hat{H}-E_{0})}[\hat{H}_{\mathrm{A}}^{\dagger},a_{j}^{\dagger}]\rangle\\
 & = & \langle a_{i}a_{j}^{\dagger}\rangle+\sum_{k}\langle a_{i}\frac{1}{z-(\hat{H}-E_{0})}a_{k}^{\dagger}\rangle h_{kj}+\sum_{k}\langle a_{i}\frac{1}{z-(\hat{H}-E_{0})}a_{k}\rangle\overline{\Delta}_{jk}\\
 & = & \langle a_{i}a_{j}^{\dagger}\rangle+\left[G^{\mathrm{hp},+}{h}\right]_{ij}+\left[G^{\mathrm{hh},+}{\Delta}^{\dagger}\right]_{ij}.
\end{eqnarray*}
 Similarly, 
\begin{eqnarray*}
zG_{ij}^{\mathrm{hp},-}(z) & = & \langle\Phi_{0}\vert(z+(\hat{H}-E_{0}))a_{j}^{\dagger}\frac{1}{z+(\hat{H}-E_{0})}a_{i}\vert\Phi_{0}\rangle\\
 & = & \langle a_{j}^{\dagger}a_{i}\rangle+\langle[\hat{H},a_{j}^{\dagger}]\frac{1}{z+(\hat{H}-E_{0})}a_{i}\rangle\\
 & = & \langle a_{j}^{\dagger}a_{i}\rangle+\langle[\hat{H}_{\mathrm{NA}},a_{j}^{\dagger}]\frac{1}{z+(\hat{H}-E_{0})}a_{i}\rangle+\langle[\hat{H}_{\mathrm{A}}^{\dagger},a_{j}^{\dagger}]\frac{1}{z+(\hat{H}-E_{0})}a_{i}\rangle\\
 & = & \langle a_{j}^{\dagger}a_{i}\rangle+\sum_{k}\langle a_{k}^{\dagger}\frac{1}{z+(\hat{H}-E_{0})}a_{i}\rangle h_{kj}+\sum_{k}\langle a_{k}\frac{1}{z+(\hat{H}-E_{0})}a_{i}\rangle\overline{\Delta}_{jk}\\
 & = & \langle a_{j}^{\dagger}a_{i}\rangle+\left[G^{\mathrm{hp},-}{h}\right]_{ij}+\left[G^{\mathrm{hh},-}{\Delta}^{\dagger}\right]_{ij}.
\end{eqnarray*}
 Therefore, adding our results and recognizing that $\langle a_{i}a_{j}^{\dagger}\rangle+\langle a_{j}^{\dagger}a_{i}\rangle=\delta_{ij}$,
we obtain 
\[
zG_{ij}^{\mathrm{hp}}=\delta_{ij}+\left[G^{\mathrm{hp}}{h}\right]_{ij}+\left[G^{\mathrm{hh}}{\Delta}^{\dagger}\right]_{ij} 
\]
for all $j>p$, which implies our first desired result.

Next compute 
\begin{eqnarray*}
zG_{ij}^{\mathrm{hh},+}(z) & = & \langle\Phi_{0}\vert a_{i}\frac{1}{z-(\hat{H}-E_{0})}a_{j}(z-(\hat{H}-E_{0}))\vert\Phi_{0}\rangle\\
 & = & \langle a_{i}a_{j}\rangle+\langle a_{i}\frac{1}{z-(\hat{H}-E_{0})}[\hat{H},a_{j}]\rangle\\
 & = & \langle a_{i}a_{j}\rangle+\langle a_{i}\frac{1}{z-(\hat{H}-E_{0})}[\hat{H}_{\mathrm{NA}},a_{j}]\rangle+\langle a_{i}\frac{1}{z-(\hat{H}-E_{0})}[\hat{H}_{\mathrm{A}},a_{j}]\rangle\\
 & = & \langle a_{i}a_{j}\rangle-\sum_{k}\langle a_{i}\frac{1}{z-(\hat{H}-E_{0})}a_{k}\rangle h_{jk}+\sum_{k}\langle a_{i}\frac{1}{z-(\hat{H}-E_{0})}a_{k}^{\dagger}\rangle\Delta_{kj}\\
 & = & \langle a_{i}a_{j}\rangle-\left[G^{\mathrm{hh},+}{h}^{T}\right]_{ij}+\left[G^{\mathrm{hp},+}{\Delta}\right]_{ij},
\end{eqnarray*}
 and 
\begin{eqnarray*}
zG_{ij}^{\mathrm{hh},-}(z) & = & \langle\Phi_{0}\vert(z+(\hat{H}-E_{0}))a_{j}\frac{1}{z+(\hat{H}-E_{0})}a_{i}\vert\Phi_{0}\rangle\\
 & = & \langle a_{j}a_{i}\rangle+\langle[\hat{H},a_{j}]\frac{1}{z+(\hat{H}-E_{0})}a_{i}\rangle\\
 & = & \langle a_{j}a_{i}\rangle+\langle[\hat{H}_{\mathrm{NA}},a_{j}]\frac{1}{z+(\hat{H}-E_{0})}a_{i}\rangle+\langle[\hat{H}_{\mathrm{A}},a_{j}]\frac{1}{z+(\hat{H}-E_{0})}a_{i}\rangle\\
 & = & \langle a_{j}a_{i}\rangle-\sum_{k}\langle a_{k}\frac{1}{z+(\hat{H}-E_{0})}a_{i}\rangle h_{jk}+\sum_{k}\langle a_{k}^{\dagger}\frac{1}{z+(\hat{H}-E_{0})}a_{i}\rangle\Delta_{kj}\\
 & = & \langle a_{j}a_{i}\rangle-\left[G^{\mathrm{hh},-}{h}^{T}\right]_{ij}+\left[G^{\mathrm{hp},-}{\Delta}\right]_{ij}.
\end{eqnarray*}
 Adding our results and recognizing that $\langle a_{i}a_{j}\rangle+\langle a_{j}a_{i}\rangle=0$,
we obtain our second desired result.

Next compute 
\begin{eqnarray*}
zG_{ij}^{\mathrm{pp},+}(z) & = & \langle\Phi_{0}\vert a_{i}^{\dagger}\frac{1}{z-(\hat{H}-E_{0})}a_{j}^{\dagger}(z-(\hat{H}-E_{0}))\vert\Phi_{0}\rangle\\
 & = & \langle a_{i}^{\dagger}a_{j}^{\dagger}\rangle+\langle a_{i}^{\dagger}\frac{1}{z-(\hat{H}-E_{0})}[\hat{H},a_{j}^{\dagger}]\rangle\\
 & = & \langle a_{i}^{\dagger}a_{j}^{\dagger}\rangle+\langle a_{i}^{\dagger}\frac{1}{z-(\hat{H}-E_{0})}[\hat{H}_{\mathrm{NA}},a_{j}^{\dagger}]\rangle+\langle a_{i}^{\dagger}\frac{1}{z-(\hat{H}-E_{0})}[\hat{H}_{\mathrm{A}}^{\dagger},a_{j}^{\dagger}]\rangle\\
 & = & \langle a_{i}^{\dagger}a_{j}^{\dagger}\rangle+\sum_{k}\langle a_{i}^{\dagger}\frac{1}{z-(\hat{H}-E_{0})}a_{k}^{\dagger}\rangle h_{kj}+\sum_{k}\langle a_{i}^{\dagger}\frac{1}{z-(\hat{H}-E_{0})}a_{k}\rangle\overline{\Delta}_{jk}\\
 & = & \langle a_{i}^{\dagger}a_{j}^{\dagger}\rangle+\left[G^{\mathrm{pp},+}{h}\right]_{ij}+\left[G^{\mathrm{ph},+}{\Delta}^{\dagger}\right]_{ij},
\end{eqnarray*}
 and 
\begin{eqnarray*}
zG_{ij}^{\mathrm{pp},-}(z) & = & \langle\Phi_{0}\vert(z+(\hat{H}-E_{0}))a_{j}^{\dagger}\frac{1}{z+(\hat{H}-E_{0})}a_{i}^{\dagger}\vert\Phi_{0}\rangle\\
 & = & \langle a_{j}^{\dagger}a_{i}^{\dagger}\rangle+\langle[\hat{H},a_{j}^{\dagger}]\frac{1}{z+(\hat{H}-E_{0})}a_{i}^{\dagger}\rangle\\
 & = & \langle a_{j}^{\dagger}a_{i}^{\dagger}\rangle+\langle[\hat{H}_{\mathrm{NA}},a_{j}^{\dagger}]\frac{1}{z+(\hat{H}-E_{0})}a_{i}^{\dagger}\rangle+\langle[\hat{H}_{\mathrm{A}}^{\dagger},a_{j}^{\dagger}]\frac{1}{z+(\hat{H}-E_{0})}a_{i}^{\dagger}\rangle\\
 & = & \langle a_{j}^{\dagger}a_{i}^{\dagger}\rangle+\sum_{k}\langle a_{k}^{\dagger}\frac{1}{z+(\hat{H}-E_{0})}a_{i}^{\dagger}\rangle h_{kj}+\sum_{k}\langle a_{k}\frac{1}{z+(\hat{H}-E_{0})}a_{i}^{\dagger}\rangle\overline{\Delta}_{jk}\\
 & = & \langle a_{j}^{\dagger}a_{i}^{\dagger}\rangle+\left[G^{\mathrm{pp},-}{h}\right]_{ij}+\left[G^{\mathrm{ph},-}{\Delta}^{\dagger}\right]_{ij},
\end{eqnarray*}
yielding our third desired result. 

Finally, compute 
\begin{eqnarray*}
zG_{ij}^{\mathrm{ph},+}(z) & = & \langle\Phi_{0}\vert a_{i}^{\dagger}\frac{1}{z-(\hat{H}-E_{0})}a_{j}(z-(\hat{H}-E_{0}))\vert\Phi_{0}\rangle\\
 & = & \langle a_{i}^{\dagger}a_{j}\rangle+\langle a_{i}^{\dagger}\frac{1}{z-(\hat{H}-E_{0})}[\hat{H},a_{j}]\rangle\\
 & = & \langle a_{i}^{\dagger}a_{j}\rangle+\langle a_{i}^{\dagger}\frac{1}{z-(\hat{H}-E_{0})}[\hat{H}_{\mathrm{NA}},a_{j}]\rangle+\langle a_{i}^{\dagger}\frac{1}{z-(\hat{H}-E_{0})}[\hat{H}_{\mathrm{A}},a_{j}]\rangle\\
 & = & \langle a_{i}^{\dagger}a_{j}\rangle-\sum_{k}\langle a_{i}^{\dagger}\frac{1}{z-(\hat{H}-E_{0})}a_{k}\rangle h_{jk}+\sum_{k}\langle a_{i}^{\dagger}\frac{1}{z-(\hat{H}-E_{0})}a_{k}^{\dagger}\rangle\Delta_{kj}\\
 & = & \langle a_{i}^{\dagger}a_{j}\rangle-\left[G^{\mathrm{ph},+}{h}^{T}\right]_{ij}+\left[G^{\mathrm{pp},+}{\Delta}\right]_{ij},
\end{eqnarray*}
 and 
\begin{eqnarray*}
zG_{ij}^{\mathrm{ph},-}(z) & = & \langle\Phi_{0}\vert(z+(\hat{H}-E_{0}))a_{j}\frac{1}{z+(\hat{H}-E_{0})}a_{i}^{\dagger}\vert\Phi_{0}\rangle\\
 & = & \langle a_{j}a_{i}\rangle+\langle[\hat{H},a_{j}]\frac{1}{z+(\hat{H}-E_{0})}a_{i}^{\dagger}\rangle\\
 & = & \langle a_{j}a_{i}\rangle+\langle[\hat{H}_{\mathrm{NA}},a_{j}]\frac{1}{z+(\hat{H}-E_{0})}a_{i}^{\dagger}\rangle+\langle[\hat{H}_{\mathrm{A}},a_{j}]\frac{1}{z+(\hat{H}-E_{0})}a_{i}^{\dagger}\rangle\\
 & = & \langle a_{j}a_{i}\rangle-\sum_{k}\langle a_{k}\frac{1}{z+(\hat{H}-E_{0})}a_{i}^{\dagger}\rangle h_{jk}+\sum_{k}\langle a_{k}^{\dagger}\frac{1}{z+(\hat{H}-E_{0})}a_{i}^{\dagger}\rangle\Delta_{kj}\\
 & = & \langle a_{j}a_{i}\rangle-\left[G^{\mathrm{ph},-}{h}^{T}\right]_{ij}+\left[G^{\mathrm{pp},-}{\Delta}\right]_{ij},
\end{eqnarray*}
 yielding the last desired result.\end{proof}

\appendix

\section{Second quantization}\label{sec:secondq}
Here we introduce the formalism of second quantization, with the aim 
of providing enough background and results to make the results of this paper rigorous.
We limit our discussion to fermionic and bosonic Fock spaces with 
\emph{finitely many} states, i.e., finitely many creation and annihilation operators. 
This setting can directly describe lattice models such as the Hubbard model 
in addition to tight-binding approximations of continuum systems. In this sense
we can 
view the set $\{1,\ldots,d\}$ as indexing sites in a lattice model.
More generally, 
one can reduce a continuum problem to this setting via the choice of a finite orbital basis \cite{NegeleOrland1988}.

\subsection{The occupation number construction}
Let $\mathcal{N}_{-1} = \{0,1\}$ 
and $\mathcal{N}_{+1} = \{0,1,2,\ldots \}$. These are the 
sets of allowable occupation numbers of a state in the fermionic and 
bosonic cases, respectively. (Recall that 
the cases $\zeta = -1$ and $\zeta = +1$ indicate, respectively, the cases 
of fermions and bosons.)

Let $d$ be a 
positive integer, the number of \emph{states}, and consider the collection: 
\[
\mathcal{B}_{\zeta,d} =  \left\{ \vert n_1, n_2, \ldots, n_d \rangle \, :\, n_i \in \mathcal{N}_\zeta \right\}.
\]
This set will be the \emph{occupation number basis} for our Fock space $\mathcal{F}_{\zeta,d}$. 
For short, we may indicate the basis elements by $\vert \mathbf{n} \rangle$ for 
 $\mathbf{n} \in  \mathcal{N}_{\zeta}^{d}$.

To define the Fock space, consider the set 
$\mathcal{V}_{\zeta,d}$ of finite formal linear combinations of elements of $\mathcal{B}_{\zeta,d}$.
Then $\mathcal{V}_{\zeta,d}$ is a vector space, and it can be endowed with an inner product 
by stipulating that the elements of $\mathcal{B}_{\zeta,d}$ are orthonormal. (Hence 
$\mathcal{B}_{\zeta,d}$ is an orthonormal basis of $\mathcal{V}_{\zeta,d}$.) For fermions,
$\mathcal{V}_{\zeta,d}$ is finite-dimensional and therefore a Hilbert space, 
but this is not the case for bosons. Therefore we define $\mathcal{F}_{\zeta,d}$ to be the completion of 
$\mathcal{V}_{\zeta,d}$ with respect to the metric induced by its inner product, so 
$\mathcal{F}_{\zeta,d}$ is a Hilbert space, and $\mathcal{B}_{\zeta,d}$ is a complete 
orthonormal set (in fact, a basis if $\zeta = -1$). 

In accordance with Dirac's bra-ket notation, we will denote elements of 
$\mathcal{F}_{\zeta,d}$ with the notation $\vert \phi \rangle$ (where 
$\phi$ can be thought of as a symbolic label), and we denote the adjoint
of an element
$\vert \phi \rangle \in \mathcal{F}_{\zeta,d}$ by $\langle \phi \vert$. 
Inner products may then be denoted $\langle \psi \vert \phi \rangle$, and 
we denote the induced norm on $\mathcal{F}_{\zeta,d}$ by 
$\Vert \phi \Vert = \sqrt{ \langle \phi \vert \phi \rangle}$.

The reader should be careful to distinguish between the \emph{vacuum state} $\vert 0,\ldots,0\rangle$, 
denoted $\vert -\rangle$ for short, and the \emph{zero vector} of $\mathcal{F}_{\zeta,d}$, denoted 
simply as $0$,\footnote{For the zero vector we forgo the bra-ket notation here to avoid confusion.} which is the linear combination of elements of $\mathcal{B}_{\zeta,d}$ in which 
all coefficients are zero. In particular the vacuum state has norm $1$ and the zero vector has norm $0$.

\subsection{Creation and annihilation operators}
The annihilation operators $a_i$ are linear operators 
$\mathcal{V}_{\zeta,d} \ra \mathcal{V}_{\zeta,d}$, defined 
by their action on the basis $\mathcal{B}_{\zeta,d}$:
\[
a_i  \vert \mathbf{n} \rangle = 
\begin{cases} 
0, & n_i =0 \\
\zeta^{\sum_{j<i} n_j} \sqrt{n_i}\  \vert n_1, \ldots, n_{i-1}, n_i - 1, n_{i+1} \ldots n_d \rangle, & n_i \neq 0.
\end{cases}
\]
Meanwhile the creation operators $a_i^\dagger$  are linear operators 
$\mathcal{V}_{\zeta,d} \ra \mathcal{V}_{\zeta,d}$, defined 
by their action on the basis $\mathcal{B}_{\zeta,d}$:
\[
a_i^\dagger  \vert \mathbf{n} \rangle = 
\begin{cases} 
0, & \zeta = -1, n_i = 1 \\
\zeta^{\sum_{j<i} n_j} \sqrt{n_i + 1}\  \vert n_1, \ldots, n_{i-1}, n_i + 1, n_{i+1} \ldots n_d \rangle, & \textrm{otherwise}.
\end{cases}
\]

In the case of fermions, $\mathcal{V}_{\zeta,d} = \mathcal{F}_{\zeta,d}$ 
is finite-dimensional, so the creation and annihilation operators are defined on 
$\mathcal{F}_{\zeta,d}$, and they are in fact Hermitian adjoints of one another in 
the usual sense as the notation 
suggests. Moreover, $a_i$ and $a_i^\dagger$ have operator norm $1$, which 
in particular remains bounded in the limit of infinitely many states. In fact, 
as this observation suggests,
the appropriate creation and annihilation operators on fermionic Fock spaces generated 
by infinitely many states (which we do not define or consider these here) 
are bounded operators with operator norm $1$.

By contrast, in the case of bosons, $a_i$ and $a_i^\dagger$ are 
\emph{unbounded} operators, even for finite $d$. Thus $a_i$ and $a_i^\dagger$ 
are only densely defined (unbounded) operators on $\mathcal{F}_{\zeta,d}$. 
In fact, adjoint operators can be defined even for operators that are only 
densely defined on a Hilbert space ~\cite{Rudin1991,GustafsonSigal2011}, and in this sense 
$a_i$ and $a_i^\dagger$ are Hermitian adjoints of one another. For the reader 
familiar only with adjoints of bounded operators, one can merely consider 
the `$\dagger$' as a notation.

The most important feature of the creation and annihilation operators 
are the (anti)commutation relations. Denoting commutator of operators $A,B$ 
by $[A,B]_{+1} := AB - BA$ 
and the anticommutator by $[A,B]_{-1} = AB + BA$, we have 
\begin{equation}
\label{eqn:commutation}
[a_i, a_j]_{\zeta} = [a_i^\dagger, a_j^\dagger]_{\zeta} = 0,\quad [a_i,a_j^\dagger]_{\zeta} = \delta_{ij} \,\mathrm{Id},
\end{equation}
on $\mathcal{V}_{\zeta,d}$. These relations can be readily verified from the definitions 
of $a_i$ and $a_i^\dagger$.

We say that a composition of creation and annihilation operators 
such as $a_i^\dagger a_j$ is \emph{normally ordered} if all of the
creation operators appear to the left of all of the annihilation operators.
Any composition of creation and annihilation operators can be converted 
to a linear combination of normally ordered operators via the (anti)commutation 
relations.

For a unitary transformation $T: \Rd \ra \Rd$, one can 
define new operators $\tilde{a}_i^\dagger = \sum_{j=1}^d T_{ij} a_j^\dagger$ 
and $\tilde{a}_i = \sum_{j=1}^d \overline{T}_{ij} a_j$. 
These can be viewed as creation and annihilation operators, respectively, in that 
they satisfy the same commutation relations as in \eqref{eqn:commutation}. 
One can in turn view these as generators for our Fock space inducing 
a different occupation number basis.

\subsection{Number operators and eigenspaces}
For each state we define a \emph{number operator}
\[
\hat{n}_i := a_i^\dagger a_i,
\]
which is a linear operator  $\mathcal{V}_{\zeta,d} \ra  \mathcal{V}_{\zeta,d}$. 
In the case of bosons $\hat{n}_i$ can be viewed as an unbounded, 
self-adjoint, densely defined operator on $\mathcal{F}_{\zeta,d}$. 
Note that the number operators all commute, i.e., $[n_i, n_j]_{+} = 0$ 
for all $i,j$.

We also 
define the \emph{total number operator}
by 
\[
\hat{N} := \sum_{i=1}^d \hat{n}_i.
\]
The set of eigenvectors of $\hat{n}$ (as a linear transformation $\mathcal{V}_{\zeta,d} \ra \mathcal{V}_{\zeta,d}$) 
is precisely $\mathcal{B}_{\zeta,d}$, and each eigenvector $\vert \mathbf{n} \rangle$ has eigenvalue 
$\sum_{i=1}^d n_i$. 
Thus the set of eigenvalues is given by $\{0,1,\ldots,d\}$ in the case of fermions and $\{0,1,\ldots \}$ 
in the case of bosons.

Then we define the \emph{$N$-particle subspace} 
to be the
$N$-eigenspace of $\hat{N}$, which is \emph{finite-dimensional} (even for bosons), 
and we denote it by $\mathcal{F}_{\zeta,d}^{(N)}$. Then we can write
\[
\mathcal{V}_{\zeta,d} = \bigoplus_{N=0}^{\infty} \mathcal{F}_{\zeta,d}^{(N)}.
\]
The $N$-eigenspace is understood to be $\{ 0\}$ for any integer $N \notin \{0,1,\ldots,d\}$ in 
the case of 
fermions and for any $N \notin \{0,1,\ldots \}$ in the case of bosons.
Notice that $a_i$ maps $\mathcal{F}_{\zeta,d}^{(N)} \ra \mathcal{F}_{\zeta,d}^{(N-1)}$ 
and $a_i^\dagger$ maps $\mathcal{F}_{\zeta,d}^{(N)} \ra \mathcal{F}_{\zeta,d}^{(N+1)}$. 

We say that an operator $A$ on $\mathcal{F}_{\zeta,d}$ 
\emph{conserves particle number} if $A$ maps  $\mathcal{F}_{\zeta,d}^{(N)} 
\ra \mathcal{F}_{\zeta,d}^{(N)}$ for all integers $N$. Evidently any operator such 
as $a_i^\dagger a_j$ in which an equal number of creation and annihilation 
operators appear, as well as any sum of such operators, must conserve particle number.

\subsection{Hamiltonians}
For convenience we shall let $a = (a_1, \ldots, a_d)$ denote the vector 
of annihilation operators, and accordingly $a^\dagger = (a_1^\dagger, \ldots, a_d^\dagger)^T$.
Then consider a \emph{Hamiltonian} $\hat{H} = H(a^\dagger, a)$ that is a  
normally ordered polynomial of creation and annihilation operators. As an operator 
on $\mathcal{V}_{\zeta,d}$, we stipulate that $\hat{H}$ commutes with the total 
number operator $\hat{N}$. Hence $\hat{H}$ 
conserves particle number and is an operator 
$\mathcal{F}_{\zeta,d}^{(N)} \ra \mathcal{F}_{\zeta,d}^{(N)}$ for all $N$, and we demand 
that $\hat{H}$ is self-adjoint as such.

In general we can write 
\[
\hat{H} = \hat{H}_0 + \hat{U},
\]
where 
\[
\hat{H}_0 := a^\dagger h a = \sum_{i,j=1}^d h_{ij} a_i^\dagger a_j
\]
is the \emph{single-particle} (or \emph{non-interacting}) part of the Hamiltonian, 
specified by a Hermitian $d\times d$ matrix $h$, and $\hat{U} = U(a^\dagger, a)$ 
is the \emph{interacting} part. Here $\hat{U}$ is normally ordered and 
commutes with $\hat{N}$ (since $\hat{H}_0$ does), and moreover $\hat{U}$ is self-adjoint on 
$\mathcal{F}_{\zeta,d}^{(N)}$ for all $N$ (since $\hat{H}_0$ is). 

Via a unitary transformation of the creation and annihilation operators, one can without loss of 
generality assume that $h$ is diagonal. However, the utility of this manipulation 
is limited outside of the non-interacting setting because such a transformation 
may complicate the representation of the interaction term $\hat{U}$.

Though we need not define $\hat{U}$ more explicitly for the purposes of this paper, 
for concreteness one might keep in mind the two-body interaction 
\begin{equation}
\label{eqn:twoBody}
\hat{U} = \frac{1}{2} \sum_{ijkl} (ij \vert U \vert kl) a_i^\dagger a_j^\dagger a_l a_k.
\end{equation}

We comment more concretely on how such a second-quantized two-body 
operator may arise from the finite dimensional approximation of a two-body potential in real space.
To construct a finite-state Fock space, 
one first replaces the single-particle Hilbert space
$\mc{H} := H^1(\R^3,\pm \frac12; \mathbb{C}) \subset L^2(\R^3, \pm \frac12; \mathbb{C})$ with a finite-dimensional 
subspace $\mc{H}_{\mc{D}}$ spanned by an orthonormal single-particle basis 
$\mc{D} := \{ \vp_1 ,\ldots, \vp_d\}$. One then defines the $N$-particle space as 
$\mc{F}^{(N)}_{\zeta,\mc{D}} := \mathbf{\Lambda}^N (\mc{H}_{\mc{D}})$ if $\zeta = -1$ and 
as $\mc{F}^{(N)}_{\zeta,\mc{D}} := \mathbf{S}^N (\mc{H}_{\mc{D}})$ if $\zeta = +1$, where 
$\mathbf{\Lambda}^N$ and $\mathbf{S}^N$ denote the $N$-th exterior and 
symmetric powers, respectively. 
Then $\mc{F}_{\zeta,\mc{D}}^{(N)}$ so constructed 
is isomorphic to $\mc{F}_{\zeta,d}^{(N)}$ as above via, in the case $\zeta = -1$, the isomorphism 
$\vert \mathbf{n} \rangle \mapsto \bigwedge_{i=1}^d \vp_{i}^{\wedge n_i}
$, where the wedge in the exponent indicates a wedge power.
The analogous isomorphism holds in the case $\zeta = +1$, with wedges 
replaced by symmetric products. {A change 
of the basis $\mc{D}$ to some $\tilde{\mc{D}} = \{ \tilde{\vp}_1 ,\ldots, \tilde{\vp}_d \}$ 
that is
induced by a unitary transformation $T: \Rd \ra \Rd$ corresponds to a transformation 
of the annihilation operators by $\tilde{a}_i = \sum_{i=1}^d \overline{T}_{ij} a_j$.} 

Under this correspondence, a translation-invariant two-body potential $v(x - y)$ in real space 
yields the tensor elements $(ij \vert U \vert kl)$ can be computed via the following 
two-body integrals \cite{NegeleOrland1988}
\begin{equation}
\label{eqn:twoBodyInt}
(ij \vert U \vert kl) = \int_{\R^3} \int_{\R^3} v(x-y) \vp_i^*(x) \vp_j^* (y) \vp_k (x) \vp_l (y)\,dx\,dy.
\end{equation}
The elements of $h$ are obtained by suitable one-body integrals; see, e.g., \cite{NegeleOrland1988}. 

In the case that $\hat{U} = U(a^\dagger, a)$ depends only on $a_i^\dagger, a_i$ 
for $i=1,\ldots, p$, we say that $\hat{H}$ is an \emph{impurity Hamiltonian}, with 
a \emph{fragment} specified by the indices $1,\ldots,p$. The rest of the indices 
correspond to the \emph{environment}.

\section{The zero-temperature ensemble}
\label{sec:zero}
At zero temperature, typically one first fixes a particle number $N$, and 
attention is restricted to the $N$-particle subspace. Let $\big\vert \Psi_0^{(N)} \big\rangle \in \mc{F}_{\zeta,d}^{(N)}$ 
be the \emph{$N$-particle ground state} of $\hat{H}$, i.e., the minimal 
normalized eigenvector of 
$\hat{H}$ considered as an operator on the $N$-particle subspace.
The role of the density operator is assumed by the 
orthogonal projector $\big\vert \Psi_0^{(N)} \big\rangle \big\langle \Psi_0^{(N)} \big\vert$ 
onto the ground state $\big\vert \Psi_0^{(N)} \big\rangle$, i.e., the 
statistical average of a linear operator $\hat{A}$ (with respect to the $N$-particle 
\emph{canonical ensemble}) is given by 
\[
\langle \hat{A} \rangle_{N} = \big\langle \Psi_0^{(N)} \big\vert \hat{A} \big\vert \Psi_0^{(N)} \big\rangle.
\]

\subsection{Green's functions and the self-energy at zero temperature}
\label{sec:zeroGreens}

For $t \in \R$, we denote the annihilation and creation operators in the Heisenberg representation by 
\begin{equation}
\label{eq:intPic}
a_i(t) := e^{i \hat{H} t} a_i e^{-i \hat{H} t}, \quad a_i^\dagger (t) := e^{i \hat{H} t} a_i^\dagger e^{-i \hat{H} t}.
\end{equation}
Then for a zero-temperature ensemble with $N$ particles, 
the \emph{time-ordered, single-body, real-time Green's function} (which we call the \emph{Green's function} for short) is 
a function $G:\R \times \R \ra \mathbb{C}^{d\times d}$ defined by 
\begin{equation}
\label{eq:rtGreen}
G_{ij} (t,t') = -i\, \big\langle \Psi_0^{(N)} \big\vert\,\mathcal{T}\big\{ a_i(t)  a_j^\dagger (t') \big\} \,\big\vert \Psi_0^{(N)} \big\rangle,  
\end{equation}
where $\mc{T}$ is the \emph{time-ordering operator}, formally defined by 
\begin{equation*}
   \mathcal{T}\big\{ a_i(t)  a_j^\dagger (t') \big\} =
   \begin{cases}
    a_{i}(t)a_{j}^{\dagger}(t'), & t' < t\\
     \zeta a_{j}^{\dagger}(t') a_{i}(t), & t' \geq t.
  \end{cases}
  \label{}
\end{equation*}
Note that $\mathcal{T}$ is not really an operator and it is 
interpreted merely via the \emph{symbolic} content of its argument.

 We can write 
 \[
 G(t,t') = G^{+} (t,t') + G^{-} (t,t'),
 \]
 where 
 \begin{equation*}
 \begin{aligned}
 i G^{+}_{ij} (t,t') := & \  \big\langle \Psi_0^{(N)} \big\vert a_{i}(t)a_{j}^{\dagger}(t') \big\vert \Psi_0^{(N)} \big\rangle \theta(t-t'), \\
 i G^{-}_{ij} (t,t') := & \  \zeta \big\langle \Psi_0^{(N)} \big\vert  a_{j}^{\dagger}(t') a_{i}(t) \big\vert \Psi_0^{(N)} \big\rangle (1-\theta(t-t')),
 \end{aligned}
 \end{equation*}
 with 
 \begin{equation}
 \label{eq:theta}
 \theta(s) := \begin{cases}
    1, & s > 0\\ 
     0, & s \leq 0.
  \end{cases}
 \end{equation}
 
 It is easy to show that $G(t,t')$, $G^{+}(t,t')$, and $G^{-}(t,t')$ depend only on $t-t'$, so we 
 can define $G(t) := G(t,0)$, $G^{+}(t): =G^{+}(t,0)$, and $G^{-}(t): =G^{-}(t,0)$ and 
 consider these objects without any loss of information. It is then equivalent to 
 consider the Fourier transforms 
 \[
 G(\omega) := \int_\R G(t) e^{i\omega t - \eta \vert t \vert}   \,dt 
 \]
 and likewise $G^{+}(\omega)$ and $G^{-}(\omega)$ defined similarly, so 
 \[
 G(\omega) = G^{+}(\omega) + G^{-}(\omega).
 \]
 Here $\eta$ is interpreted as a positive, infinitesimally small quantity needed 
 to ensure the convergence of the relevant integrals, and $G(\omega)$, $G^{+}(\omega)$, and $G^{-}(\omega)$ 
 are not really functions, but rather distributions on $\R$ defined via the limit $\eta \ra 0^+$.
 
One can show that 
\[
G_{ij}^{+} (\omega) = \big\langle \Psi_0^{(N)} \big\vert a_i \frac{1}{\omega - (\hat{H} - E_0^{(N)}) + i \eta} a_j^{\dagger}  \big\vert  \Psi_0^{(N)} \big\rangle
\]
and
\[
G_{ij}^{-} (\omega) = -\zeta \big\langle \Psi_0^{(N)} \big\vert a_j^{\dagger} \frac{1}{\omega + (\hat{H} - E_0^{(N)}) - i \eta} a_i  \big\vert  \Psi_0^{(N)} \big\rangle,
\]
where $E_0^{(N)}$ is the energy of the $N$-particle ground state, i.e., $\hat{H}  \big\vert  \Psi_0^{(N)} \big\rangle = E_0  \big\vert  \Psi_0^{(N)} \big\rangle$.

Now we can think of $G^{\pm}$ as the restriction to the real axis of the rational function 
$G^{\pm} : \mathbb{C} \ra \mathbb{C}^{d\times d}$ defined by 
\begin{equation*}
\begin{aligned}
G_{ij}^{+} (z) := & \ \big\langle \Psi_0^{(N)} \big\vert a_i \frac{1}{z - (\hat{H} - E_0^{(N)}) } a_j^{\dagger}  \big\vert  \Psi_0^{(N)} \big\rangle \\
G_{ij}^{-} (z) := & \ -\zeta \big\langle \Psi_0^{(N)} \big\vert a_j^{\dagger} \frac{1}{z + (\hat{H} - E_0^{(N)}) } a_i  \big\vert  \Psi_0^{(N)} \big\rangle,
\end{aligned}
\end{equation*}
and we can define $G(z) := G^{+}(z) + G^{-}(z)$ accordingly to be rational on $\mathbb{C}$.

Note that here we have left out the $\pm i \eta$ in the denominators, which specified whether 
poles should be viewed as being infinitesimally above or below the real axis. This erases 
the distinction between the time-ordered Green's function and the advanced and retarded 
Green's functions, which we do not define here, though see \cite{NegeleOrland1988} for
details. In fact 
the distinction does not matter for our sparsity results, which applies equally well in all of these 
cases. 

The self-energy is the rational function $\Sigma: \mathbb{C} \ra \mathbb{C}^{d\times d}$ defined by 
\[
\Sigma(z) := z - h - G(z)^{-1}.
\]

\section{The finite-temperature ensemble}
\label{sec:finite}
At inverse temperature $\beta \in (0,\infty)$, the
\emph{partition function}
is defined by 
\begin{equation}
\label{eq:partition}
Z := \Tr[e^{-\beta (\hat H - \mu \hat{N})}].
\end{equation}
where `$\Tr$' indicates the Fock space trace. 
Here $\mu \in \R$ is the \emph{chemical potential}, 
but before commenting on its role,
some further 
elaboration on the trace is owed in the bosonic case, in which 
the Fock space is infinite-dimensional.

By assumption, $\hat{H}$ conserves particle number, i.e., it maps
$\mc{F}_{\zeta,d}^{(N)}$ to itself for all $N$. Thus $e^{-\beta (\hat H - \mu \hat{N})}$ does as well 
and can be viewed as 
a positive-definite operator on each $\mc{F}_{\zeta,d}^{(N)}$.
The trace can then be constructed as 
\[
\Tr[e^{-\beta (\hat H - \mu \hat{N})}] = 
\sum_{N=0}^\infty 
\Tr_{N} [e^{-\beta (\hat H - \mu \hat{N})}] = 
\sum_{N=0}^\infty e^{\beta \mu N} \,
\Tr_{N} [e^{-\beta \hat H}], 
\]
where `$\Tr_N$' indicates the 
trace on the $N$-particle subspace. Since each of the summands is positive, 
$\Tr[e^{-\beta (\hat H - \mu \hat{N})}] \in (0,+\infty]$ is well-defined. 

More generally, the trace is defined for all operators in the trace class of 
$\mc{F}_{\zeta,d}$, i.e., the set of bounded linear 
operators $\hat{O}:\mc{F}_{\zeta,d} \ra \mc{F}_{\zeta,d}$ 
for which 
\[
\sum_{ \mathbf{n} \in \mc{N}_{\zeta}^d} \langle \mathbf{n} \vert \,(\hat{O}^\dagger \hat{O})^{1/2} \, \vert \mathbf{n} \rangle < +\infty,
\]
in which case 
\[
\Tr[\hat{O}] = \sum_{ \mathbf{n} \in \mc{N}_{\zeta}^d} \langle \mathbf{n} \vert \hat{O} \vert \mathbf{n} \rangle.
\]
 See, e.g., \cite{ReedSimon1980} for further details on trace class operators.

Now since the partition function can be viewed as a normalization factor, 
the scenario $Z = +\infty$ is to be avoided. It is now that we turn to the
chemical potential. We can view $Z$ as defined above as a function of 
$\mu$. Evidently $\mu \mapsto Z(\mu)$ is non-decreasing.

First we want to rule out the case that $Z \equiv +\infty$. Unfortunately, 
this case cannot be ruled out without further assumptions! To see why, 
suppose that $d = 1$ (so write $a = a_1$), and let $\hat{H} = -a^\dagger a - a^\dagger a^\dagger a  a =  - a^\dagger a a^\dagger a = - \Hat{N}^2$. 
Then 
\[
\Tr[e^{-\beta (\hat H - \mu \hat{N})}] = 
\sum_{N=0}^\infty 
e^{\beta (N^2 + \mu N)}\, \Tr_{N} \left[\mathrm{Id}_{\mc{F}_{\zeta,d}^{(N)}}\right] 
= \sum_{N=0}^\infty  e^{\beta (N^2 + \mu N)}\,\binom{N + d-1}{d-1}
= +\infty, 
\]
for \emph{all} $\mu \in \R$.

We conclude that such a choice of $\hat{H}$ is \emph{unphysical}, and 
to rule out such pathologies, we adopt the following: 
\begin{assumption}
\label{assumption:bosons}
We assume, in the case of bosons, that there exist some positive integer 
$N_0$ and some $\mu \in \R$ such that 
$\hat{H} - \mu \hat{N} \succeq 0$ 
as an operator on all $N$-particle subspaces for $N\geq N_0$. 
(It is equivalent to require 
that there exist $N_0, \mu$ 
such that $\hat{U} - \mu \hat{N} \succeq 0$ on all $N$-particle subspaces 
for $N\geq N_0$.)
\end{assumption}

This condition is satisfied 
in particular if $\hat{U}$ is a two-body interaction as in \eqref{eqn:twoBody}
such that $U_{ik,jl} := (kj\vert U\vert il)$, interpreted as a $d^2 \times d^2$ 
matrix, is positive semidefinite. 
Indeed, in this case, $\hat{U}$ is equal to
(up to a single-body term) 
\[
 \frac{1}{2} \sum_{ijkl} U_{ik,jl} \left[a_i^\dagger a_k\right]^\dagger \left[a_j^\dagger a_l\right] \succeq 0.
\]
If the $(ij\vert U\vert kl)$ are derived from a real-space two-body potential $v$ that 
is a positive semidefinite function (i.e., has nonnegative Fourier transform), 
then it follows from \eqref{eqn:twoBodyInt} that the matrix
$(U_{ik,jl})$ is positive definite as desired. Note that it is possible for $v$ to 
be positive definite but take negative values at long ranges, i.e., $v$ can act attractively 
at long range.

Now that we have argued that Assumption \ref{assumption:bosons} is 
natural, let us see how it guarantees that the 
set $\dom Z := \{ \mu \,:\, Z(\mu) < +\infty\}$ is non-empty. Indeed, choose $\mu'$ negative
enough such that $\hat{H}- \mu' \hat{N} \succeq 0$ as an operator on all 
$N$-particle subspaces, and let $\mu = \mu' - \delta$, where $\delta >0$. 
Then 
\[
\Tr[e^{-\beta (\hat H - \mu \hat{N})}] \leq  
\sum_{N=0}^\infty 
\Tr_{N} [e^{-\beta \delta \hat{N}}] 
= \sum_{N=0}^\infty  e^{-\beta \delta N} \binom{N + d-1}{d-1}.
\]
Now the binomial coefficient in the last expression is $O(N^{d-1})$ as $N\ra +\infty$, so the 
sum converges.

We will always assume in the finite-temperature setting 
that $\mu \in \intdom Z$.
It can be shown that if $\hat{U} = 0$, then  
$\dom Z = \{ \mu \,:\, h \succ \mu \,I_d \}$. Moreover, if there exist 
$N_0, \delta > 0$ such that 
$\hat{U} \succeq \delta \hat{N}^2$ on all $N$-particle subspaces for $N \geq N_0$ 
(which holds in particular if $\hat{U}$ is is a two-body interaction as in \eqref{eqn:twoBody}
where the $d^2 \times d^2$ matrix $U_{ki,jl} := (ij\vert U\vert kl)$ is \emph{positive definite}), 
then $\dom Z = \R$.

Notice that if $\dom Z$ is open, then since $Z$ is increasing we can write 
$\dom Z = (-\infty, \mu_{\mathrm{c}})$ for some $\mu_{\mathrm{c}}$ possibly 
infinite. If $\mu_{\mathrm{c}} < +\infty$, then by Fatou's lemma we have that
$\liminf_{\mu\ra\mu_{\mathrm{c}}^-} Z(\mu) \geq Z(\mu_{\mathrm{c}}) = +\infty$, 
so $Z(\mu) \ra +\infty$ as $\mu \ra \mu_{\mathrm{c}}^-$. 
(And in any case it follows from the definition of $Z$ that $Z(\mu) \ra +\infty$ 
as $\mu \ra +\infty$, so we can write more 
compactly that $Z(\mu) \ra +\infty$ as $\mu \ra  \mu_{\mathrm{c}}$, no matter 
whether $\mu_{\mathrm{c}}$ is finite or infinite.)

The  \emph{grand canonical ensemble} is defined by the \emph{density operator}
\[
\rho := Z^{-1} e^{-\beta (\hat H - \mu \hat{N})}, 
\]
and the statistical average of an operator $\hat{A}$ with respect to the grand canonical 
ensemble is denoted 
\[
\langle \hat{A} \rangle_{\beta,\mu} = \Tr [ \hat{A} \rho] 
\]
 whenever $\hat{A} \rho$ is in the trace class. Note that if $\hat{A}$ conserves 
 particle number then 
\[
\Tr[\hat{A} \rho] = 
\sum_{N=0}^\infty \Tr_{N} [\hat{A} \rho] = 
Z^{-1} \sum_{N=0}^\infty e^{\beta \mu N} \,
\Tr_{N} [\hat{A} e^{-\beta \hat{H}}] 
\]
is defined under the condition that the sum is absolutely convergent, which holds in 
particular if the norm of $\hat{A}$ as an operator on the $N$-particle subspace 
grows only polynomially with $N$, 
via the assumption that $\mu \in \intdom Z$.
When the context is clear 
we simply write $\langle \, \cdot \, \rangle$.

Of particular interest is the \emph{expected particle number} 
\[
\langle \hat{N} \rangle = \frac{\sum_{N=0}^\infty N e^{\beta \mu N} \,
\Tr_{N} [e^{-\beta \hat{H}}] }
{\sum_{N=0}^\infty e^{\beta \mu N} \,
\Tr_{N} [e^{-\beta \hat{H}}] }.
\] 
Observe that $\langle \hat{N} \rangle_{\beta,\mu} \ra 0$ as $\mu \ra -\infty$. Also note
that if $\dom Z = \R$, then
$\langle \hat{N} \rangle_{\beta,\mu} \ra +\infty$. 
Defining the \emph{free energy} $\Omega(\mu) := \beta^{-1} \log Z(\mu)$, we see that 
$\langle \hat{N} \rangle_{\beta,\mu} = \Omega'(\mu)$.

It is not hard to check that $\Omega$ is (strictly) convex, i.e., $\langle \hat{N} \rangle_{\beta,\mu}$ 
is increasing in $\mu$ for $\mu \in \intdom Z$. Recall that if $\dom Z = (0,\mu_{\mathrm{c}})$,
then $Z(\mu) \ra +\infty$ as $\mu \ra \mu_{\mathrm{c}}$, hence 
$\Omega(\mu) \ra +\infty$ as $\mu \ra \mu_{\mathrm{c}}$. If $\mu_{\mathrm{c}} < +\infty$, 
it follows that $\Omega'(\mu) \ra +\infty$ as $\mu \ra \mu_{\mathrm{c}}^-$. (Otherwise, 
since $\Omega'$ is increasing, it approaches a finite limit $\mu \ra \mu_{\mathrm{c}}^-$. But 
in this case it would follow from the fundamental theorem of calculus 
that $\Omega$ approaches a finite limit as well: contradiction.) In summary we 
have established that if $\dom Z$ is open, then $Z(\mu) \ra +\infty$ as $\mu \ra \mu_{\mathrm{c}}$, 
no matter whether $\mu_{\mathrm{c}}$ is finite or infinite. It follows that in this case 
$\mu \mapsto \langle \hat{N}_{\beta,\mu} \rangle$ is a bijection from $\dom Z = (-\infty, \mu_{\mathrm{c}})$ 
to $(0,+\infty)$. Thus one can \emph{select} the chemical potential $\mu$ to 
yield a predetermined expected particle number.

\subsection{Green's functions and the self-energy at finite temperature}
\label{sec:finiteGreens}

Recall our definition~\eqref{eq:intPic} of the annihilation and creation operators $a_i(t)$ and $a_i^\dagger (t)$ in the Heisenberg representation. 
%\[
%a_i(t) := e^{i \hat{H} t} a_i e^{-i \hat{H} t}, \quad a_i^\dagger (t) := e^{i \hat{H} t} a_i^\dagger e^{-i \hat{H} t}.
%\]
Then at finite inverse temperature $\beta \in (0,\infty)$ and chemical potential $\mu \in \intdom Z$, 
the \emph{time-ordered, single-body, real-time Green's function} (which we call the \emph{Green's function} for short 
when the context is clear) is 
a function $G:\R \times \R \ra \mathbb{C}^{d\times d}$ defined by 
\[
G_{ij} (t,t') = -i\, \big\langle\, \mathcal{T}\big\{ a_i(t)  a_j^\dagger (t') \big\} \,\big\rangle_{\beta,\mu}. 
\]
 
 We can write 
 \[
 G(t,t') = G^{+} (t,t') + G^{-} (t,t'),
 \]
 where 
 \begin{equation*}
 \begin{aligned}
 i G^{+}_{ij} (t,t') = &\  \frac{1}{Z} \Tr\left[ a_{i}(t)a_{j}^{\dagger}(t') e^{-\beta(\hat{H} - \mu \hat{N})}\right] \theta(t-t'), \\ 
 i G^{-}_{ij} (t,t') = &\ \frac{\zeta}{Z}  \Tr\left[ a_{j}^{\dagger}(t') a_{i}(t) e^{-\beta(\hat{H} - \mu \hat{N})}\right] (1-\theta(t-t')),
 \end{aligned}
 \end{equation*}
 where $\theta$ is defined as above in~\eqref{eq:theta}.
% \begin{equation*}
% \theta(s) := \begin{cases}
%    1, & s > 0\\ 
%     0, & s \leq 0.
%  \end{cases}
% \end{equation*}
% as above.
 
 Once again it is easy to show that $G(t,t')$, $G^{+}(t,t')$, and $G^{-}(t,t')$ depend only on $t-t'$, so we 
 can define $G(t) := G(t,0)$, $G^{+}(t): =G^{+}(t,0)$, and $G^{-}(t): =G^{-}(t,0)$ and 
 consider these objects without any loss of information. It is then equivalent to 
 consider the Fourier transforms 
 \[
 G(\omega) := \int_\R G(t) e^{i\omega t - \eta \vert t\vert }   \,dt 
 \]
 and likewise $G^{+}(\omega)$ and $G^{-}(\omega)$ defined similarly, so 
 \[
 G(z) = G^{+}(\omega) + G^{-}(\omega).
 \]
 
 Now since $\hat{H}$ preserves particle number, we can safely diagonalize $\hat{H}$ 
 as an operator on each of the $N$-particle subspaces separately. Then the spectrum of
 $\hat{H}$ consists of the union of its spectra on the $N$-particle subspaces. 
 It follows from Assumption \ref{assumption:bosons} that $\hat{H} - \mu \hat{N}$ has a ground state, i.e., that
 its spectrum is bounded from below, for $\mu \in \intdom Z$.
 Let $m=0,1,\ldots, $ (terminating at $m=2^d$ in the case of fermions) index the 
 spectrum of $\hat{H}$, and let $\vert \Psi_m \rangle$ denote the $m$-th eigenstate. Let $N_m$ 
 be the particle number of $\vert \Psi_m \rangle$ (which is an eigenstate of $\hat{N}$), and let 
 $E_m$ be defined by $ \hat{H} \vert \Psi_m \rangle = E_m \vert \Psi_m \rangle$.
 
One can show that 
\[
G_{ij}^{+} (\omega) = \frac{1}{Z} \sum_m e^{-\beta (E_m - \mu N_m)} \big\langle \Psi_m \big\vert a_i \frac{1}{\omega - (\hat{H}- E_m) + i \eta} a_j^{\dagger}  \big\vert \Psi_m \big\rangle
\]
and
\[
G_{ij}^{-} (\omega) = \frac{-\zeta}{Z} \sum_m e^{-\beta (E_m - \mu N_m)} \big\langle \Psi_m \big\vert a_j^{\dagger} \frac{1}{\omega + (\hat{H}- E_m) - i \eta} a_i  \big\vert \Psi_m \big\rangle.
\]
Recall from~\eqref{eq:partition} that 
\[
Z = \sum_m e^{-\beta (E_m - \mu N_m)}.
\]

Now we can think of $G^{\pm}$ as the restriction to the real axis of the rational function 
$G^{\pm} : \mathbb{C} \ra \mathbb{C}^{d\times d}$ defined by 
\begin{equation*}
\begin{aligned}
G_{ij}^{+} (z) := & \ \frac{1}{Z} \sum_m e^{-\beta (E_m - \mu N_m)} \big\langle \Psi_m \big\vert a_i \frac{1}{z - (\hat{H}- E_m) } a_j^{\dagger}  \big\vert \Psi_m \big\rangle \\
G_{ij}^{-} (z) := & \ \frac{-\zeta}{Z} \sum_m e^{-\beta (E_m - \mu N_m)} \big\langle \Psi_m \big\vert a_j^{\dagger} \frac{1}{z + (\hat{H}- E_m) } a_i  \big\vert \Psi_m \big\rangle,
\end{aligned}
\end{equation*}
and we can define $G(z) := G^{+}(z) + G^{-}(z)$ accordingly to be rational on $\mathbb{C}$. 
Once again we have ignored the infinitesimal $\eta$ in this definition; the same comments made 
in Appendix \ref{sec:zero} apply here.

The self-energy is the rational function $\Sigma: \mathbb{C} \ra \mathbb{C}^{d\times d}$ defined by 
\[
\Sigma(z) := z - h - G(z)^{-1}.
\]

  \section{Non-equilibrium setting and the Kadanoff-Baym contour}
  \label{sec:non-equilibrium}
  Here we briefly discuss one main non-equilibrium setting of interest,
  called the \emph{Kadanoff-Baym} formalism.  One considers an initial time $t_0$ and 
  a final time $t_1$, with $t_1 > t_0$, and for $t \in [t_0,t_1]$, $\hat{H}(t)$ denotes the Hamiltonian 
  at time $t$. This Hamiltonian determines the evolution, starting at time $t_0$, of a prepared
   grand canonical ensemble defined by a density operator $\rho$, i.e., a positive semi-definite operator on the Fock space of unit trace. Assuming, for simplicity, strict positive definiteness, we can write 
   \[
   \rho = \frac{1}{\Tr[e^{-\beta \overline{H}}]} e^{- \beta \overline{H}}
   \]
    for some Hamiltonian $\overline{H}$ and inverse temperature $\beta$. Of course, this form leaves freedom in choosing $\beta$, but it is good to think of $\beta$ as a free parameter. Often $\overline{H}$ may be thought of as 
    $\hat{H}(t_0) - \mu \hat{N}$, but this need not be the case. To ensure that 
    Assumption \ref{assumption:traceClass} holds, it will suffice
     to assume that $\Tr[e^{-\beta \overline{H} + \ve \hat{N} }] < +\infty$ 
    for some $\ve > 0$ sufficiently small. 
    This condition is analogous to the condition 
    $\mu \in \intdom Z$ discussed in Appendix \ref{sec:finite} for the equilibrium finite-temperature ensemble. 
    Assuming the condition, let $\hat{O}_N$ 
    denote the restriction of $e^{-\beta \overline{H}}$ to the $N$-particle subspace. 
    Then it follows that $\Tr [\hat{O}_N]$ decays exponentially in $N$, hence 
    $\Vert \hat{O}_N \Vert_2$ does as well.

  Here the contour is the Kadanoff-Baym contour $\mathcal{C}^{\mathrm{KB}}$, specified by the path $\gamma^{\mathrm{KB}}$, 
  which can be written as a concatenation 
  \[
  \gamma^{\mathrm{KB}} = \gamma^{-} + \gamma^{+} + \gamma^{\mathrm{M}}.
  \]
  Here $\gamma^{-}:(0,t_1-t_0) \ra \mathbb{C}$ is defined by $s \mapsto s + t_0$, 
   $\gamma^{+}:(0,t_1-t_0) \ra \mathbb{C}$ is defined by $s \mapsto t_1 - s$, and 
    $\gamma^{\mathrm{M}}:(0,\beta) \ra \mathbb{C}$ is defined by $s \mapsto t_0 - is$. 
      Accordingly we define sub-contours, $\mathcal{C}_{\pm}$ and $\mathcal{C}_{\mathrm{KB}}$.
       The concatenation $\gamma^{\mathrm{KB}}$ is viewed as a function 
   $(s_0, s_1) \ra \mathbb{C}$, where $s_0 = 0$ and $s_1 = 2(t_1 - t_0) + \beta$.
    
  We have already defined the contour Hamiltonian $\hat{H}(z)$ for $z \in \mathcal{C}_{\pm}$. To 
  complete the specification of our ensemble we stipulate that $\hat{H}(z) = \overline{H}$ for 
  $z \in \mathcal{C}_{\mathrm{M}}$ . 
  For contour times $s,s' < t_1 - t_0$, the contour-ordered Green's function $G(s,s')$ recovers the 
  appropriate notion of the 
  real-time-ordered non-equilibrium Green's function; similarly, appropriate 
  notions of advanced and retarded Green's functions 
  can be recovered from the contour-ordered Green's function. However, only the contour-ordered 
  Green's function admits a favorable perturbation theory, and this remarkable fact is one motivation
  for considering it. See \cite{StefanucciVanLeeuwen2013} for further details. In this work we additionally see that 
  the contour-ordered setting is also the natural setting in which to recover a sparsity result for the self-energy of
  impurity problems in the non-equilibrium setting.

  Now one can readily check that the partition function is given by 
  $Z = \Tr[e^{-\beta \overline{H}}] > 0$ (so Assumption \ref{assumption:Znonzero} is satisfied). 
  Now we verify Assumption \ref{assumption:traceClass}. For $s' \leq s \leq s_1 - \beta$, note that
  $U(s,s')$ is unitary, hence bounded. Moreover, for $s_1 - \beta \leq s' \leq s$, we have  
  $U(s,s') = e^{- (s-s') \overline{H}}$, which is trace class (by our assumption), hence bounded.
  It follows that for any $s_0 \leq s' \leq s \leq s_1$, the operator $U(s,s')$ is bounded.
  In fact, 
  $U( s_1, s_1 - \beta) = e^{-\beta \overline{H}}$, and as mentioned above, the
  operator norm of this 
  operator restricted to the $N$-particle subspace decays exponentially in $N$. Thus 
  Assumption \ref{assumption:traceClass} is satisfied.

\bibliographystyle{siam}
\bibliography{sparseref}

\end{document}